\documentclass[preprint,11pt]{elsarticle}

\def\tsc#1{\csdef{#1}{\textsc{\lowercase{#1}}\xspace}}
\tsc{WGM}
\tsc{QE}
\tsc{EP}
\tsc{PMS}
\tsc{BEC}
\tsc{DE}

\usepackage{amsmath,amsfonts}
\usepackage{array}
\usepackage{ntheorem}
\usepackage[caption=false,font=normalsize,labelfont=sf,textfont=sf]{subfig}
\usepackage{textcomp}
\usepackage{multirow}
\usepackage{xcolor}
\usepackage{stfloats}
\usepackage{url}
\usepackage{enumitem}
\usepackage{verbatim}
\usepackage{bm}
\usepackage{booktabs}
\usepackage{graphicx}
\usepackage{subfig}
\newtheorem{theorem}{Theorem}
\newtheorem{theorem2}{Theorem}

\newtheorem{lemma}[theorem]{Lemma}
\newtheorem{lemma2}[theorem2]{Lemma}
\newtheorem{definition}{Definition}
\newtheorem{proposition}[theorem]{Proposition}

\newtheorem*{proof}{Proof}
\newtheorem{proof2}{Proof}

\usepackage[normalem]{ulem}
\usepackage{caption}
\usepackage{algorithm} 
\usepackage{algorithmic}

\usepackage[algo2e]{algorithm2e} 
\useunder{\uline}{\ul}{}

\begin{document}
\let\WriteBookmarks\relax
\def\floatpagepagefraction{1}
\def\textpagefraction{.001}



\title{Debiased Model-based Interactive Recommendation}                      

\tnotetext[1]{This research was supported in part by  the National Key R\&D Program of China (2021ZD0111501), National Science Fund for Excellent Young Scholars (62122022), Natural Science Foundation of China (61876043, 61976052), the major key project of PCL (PCL2021A12).}


%
\author[1,4]{Zijian Li}

\ead{leizigin@gmail.com}

\affiliation[1]{organization=Guangdong University of Technology,
    addressline={School of Computer Science}, 
    city={Guangzhou City},
    postcode={510006}, 
    country={China}}

\affiliation[2]{organization=Shantou University,
    addressline={College of Science}, 
    city={Shantou City},
    postcode={515063}, 
    country={China}}

\affiliation[3]{organization=Huawei Technology,
    city={Shenzhen City},
    country={China}}

\affiliation[4]{organization={Mohamed bin Zayed University of Artificial Intelligence},
    city={Masdar City, Abu Dhabi},
    country={United Arab Emirates}}

\author[1]{Ruichu Cai*}
\ead{cairuichu@gmail.com}

\author[1]{Haiqin Huang}
\ead{huanghaiqin-hhq@gmail.com}

\author[1]{Sili Zhang}
\ead{huanghaiqin-hhq@gmail.com}

\author[1]{Yuguang Yan}
\ead{huanghaiqin-hhq@gmail.com}

\author[2]{Zhifeng Hao}
\ead{haozhifeng@stu.edu.cn}

\author[3]{Zhenhua Dong}
\ead{haozhifeng@stu.edu.cn}





\cortext[cor1]{Corresponding author}



\begin{abstract}
Existing model-based interactive recommendation systems are trained by querying a world model to capture the user preference, but learning the world model from historical logged data will easily suffer from bias issues such as popularity bias and sampling bias. This is why some debiased methods have been proposed recently. 
However, two essential drawbacks still remain: 1) ignoring the dynamics of the time-varying popularity results in a false reweighting of items. 2) taking the unknown samples as negative samples in negative sampling results in the sampling bias. To overcome these two drawbacks, we develop a model called \textbf{i}dentifiable \textbf{D}ebiased \textbf{M}odel-based \textbf{I}nteractive \textbf{R}ecommendation (\textbf{iDMIR} in short). In iDMIR, for the first drawback, we devise a debiased causal world model based on the causal mechanism of the time-varying recommendation generation process with identification guarantees; for the second drawback, we devise a debiased contrastive policy, which coincides with the debiased contrastive learning and avoids sampling bias. Moreover, we demonstrate that the proposed method not only outperforms several latest interactive recommendation algorithms but also enjoys diverse recommendation performance.
\end{abstract}



\begin{keyword}
Recommendation System \sep Causal Learning \sep Reinforcement Learning \sep Identification
\end{keyword}

\maketitle

\section{Introduction}
The interactive recommendation system \cite{zhou2020interactive,10.1145/3397271.3401174,9837877}, which is trained with the historical logged data,
\textcolor{black}{is} a special type of interactive information retrieval \cite{10.1145/3397271.3401424} and is naturally suitable for being \textcolor{black}{modeled} by reinforcement learning. One of the most predominant methods is the \textit{model-based method}, which employs a world model \cite{ha2018world} to simulate the user behaviors from the logged data, avoiding the intractable \textit{Deadly Triad} \cite{zhang2021breaking}, i.e. the instability and divergence under the offline \textcolor{black}{scenario}.

However, model-based interactive recommendation algorithms \cite{ie2019recsim,shi2019pyrecgym,zhao2019toward} usually suffer from biased estimation because the world model is trained on the historical logged data with different types of bias such as the popularity bias. As a result, a biased world model generates biased rewards for policy and consequently results in suboptimal recommendation results. 
Inspired by the power of causal effect in debiasing, recent advances \textcolor{black}{in} model-based interactive recommendation methods \cite{10.1145/3383313.3412252} leverage the technologies of causal effect like sample reweighting (e.g. reducing the weights of the popular items) or Inverse Propensity Scoring (IPS) to build a debiased world model.

Despite their general efficacy in building a debiased world model, these methods may be still constrained by two bottlenecks. First, existing model-based interactive recommendation algorithms usually ignore the dynamic of popularity as shown in Figure \ref{fig:motivation}(a). For example, \textcolor{black}{pumpkins and Christmas trees have \textcolor{black}{a} similar aggregated mean of popularity over time}, but their popularity is different on different \textcolor{black}{holidays (Christmas and Thanksgiving Day occur at the peak of the popularity in the pink line and green line respectively)}. If we straightforwardly employ the sample reweighting methods, due to the static aggregated mean of pumpkins' popularity, the weights of pumpkins will be overly reduced on Christmas, which is unreasonable. 
Second, as shown in Figure \ref{fig:motivation} (c), the negative sampling in recommendation might take the unknown samples as the negative samples \cite{zhang2021causal}, which results in the sampling bias. Since the sampling bias might falsely take some unpopular but positive items as negative items, it indirectly amplifies the popularity bias. Therefore, these two drawbacks result in bias estimation, i.e., overly recommending the Christmas tree at $t+1$, as shown in Figure \ref{fig:motivation} (b). Besides the drawbacks above, the existing methods can rarely handle the case with latent confounders

\begin{figure*}[t]
		\centering\		\includegraphics[width=\columnwidth]{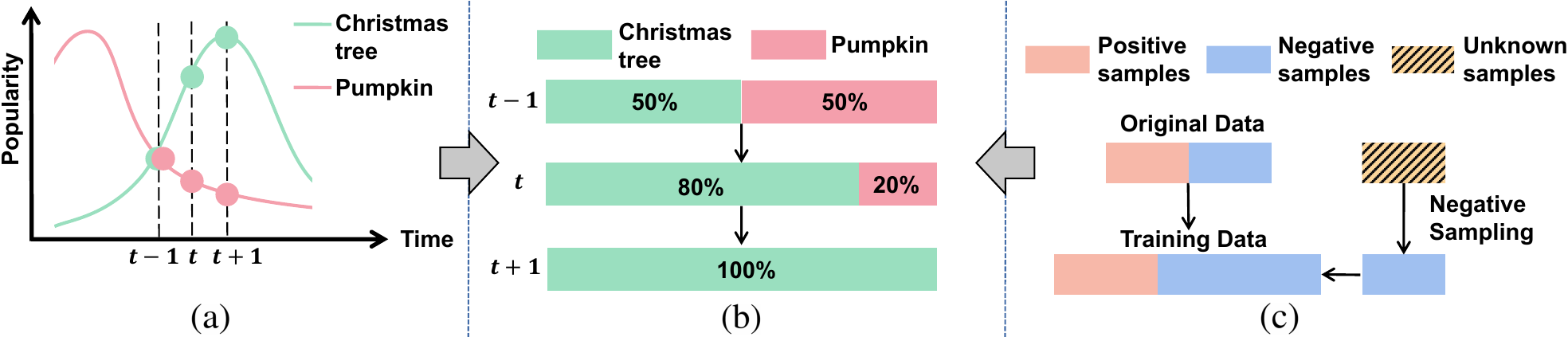}
		\caption{The illustration of two drawbacks that lead to biased estimation. (a) The popularity of items changes over time. (b) The estimated probability of different items is gradually degenerated by ignoring the dynamic popularity and sampling bias. (c) The negative sampling technique, which uses unknown samples as negative samples, will result in sampling bias. (\textit{Best view in color.})}
		\label{fig:motivation}
\end{figure*}

To overcome the aforementioned drawbacks, we develop the \textbf{i}dentifiable \textbf{D}ebiased
\textbf{M}odel-based \textbf{I}nteractive \textbf{R}ecommendation model (\textbf{iDMIR} in short). \textcolor{black}{First, we build a causal generation process \textcolor{black}{that} simultaneously considers the latent user preference confounders and time-varying popularity. To model the dynamics of popularity, we consider the popularity in different successive time intervals. Moreover, different from other debiased methods that require all observed confounders, we break this limitation by \textcolor{black}{modeling} the latent user variables via stochastic variational inference (SVI) \cite{louizos2017causal,kingma2013auto} with identification guarantees.} Second, \textcolor{black}{to avoid the sampling bias}, we devise a debiased \textcolor{black}{contrastive} policy. By separately processing the positive and negative sequences, the debiased \textcolor{black}{contrastive} policy can achieve the goal of debiasing and \textcolor{black}{aggregating} the user states. We also provide the theoretical analysis of intervention distribution and the latent variables to show the unbiasedness of the world model and policy. Extensive experimental studies demonstrate that the proposed iDMIR approach outperforms several latest interactive and debiased recommendation algorithms. 

\section{Related Works}
In this section, we first review the works about interactive recommendation systems based on reinforcement learning and debiased recommendation systems. Then we review the works about the identification of latent variables in causal models.

\subsection{Interactive Recommendation Systems}
Many \textcolor{black}{researchers} leverage 
reinforcement learning \cite{xiao2021general,kidambi2020morel,chen2021survey,afsar2022reinforcement,ahmadian2023rderl} to address the interactive recommendation problem, which can be categorized into \textcolor{black}{the} \textit{\textcolor{black}{bandit-based} methods}, \textit{model-free methods} and \textcolor{black}{the} \textit{model-based method}. 
The \textcolor{black}{bandit-based} methods \cite{liu2020diversified,li2010contextual} use the contextual \textcolor{black}{bandit} for building interactive recommendation systems. 
Li et.al \cite{li2010contextual} propose the contextual \textcolor{black}{bandit-based} method that sequentially selects articles to serve users based on contextual information. Recently, Song et.al \cite{10.1145/3397271.3401174} use the natural exploration in bandit to solve the click-through rate (CTR) underestimation problem. And Yang et.al \cite{yang2020hierarchical} propose HATCH to conduct the policy learning of contextual bandits with a budget constraint. 
The model-free methods \cite{ dulac2015deep,zhao2018deep,zou2019reinforcement} employ the ideas of model-free reinforcement learning like Monte Carlo (MC) \cite{lazaric2007reinforcement} and Temporal Difference (TD) \cite{tesauro1995temporal}. Chen et.al \cite{chen2019large} propose a tree-structured RL framework where items are sought from a balanced tree. Teng et.al \cite{xiao2021general} propose a general offline framework. Chen et.al \cite{10.1145/3442381.3449846} \textcolor{black}{use} the deterministic policy gradient \cite{silver2014deterministic} to model the multi-aspect \textcolor{black}{preference} of users.
\textit{Model-based methods} \cite{zou2020pseudo, 10.1145/3503181.3503203,nair2018overcoming} simulate the behaviors of customers and sequentially learn a policy via planning. 
Zhao et.al \cite{zhao2021usersim} develop a user simulator based on Generative Adversarial Networks. Aiming to devise an \textcolor{black}{unbiased} simulator, Zou et.al \cite{zou2020pseudo} leverage the sample reweighting technique to correct the discrepancy between recommendation policy and logging policy, and Huang et.al \cite{10.1145/3383313.3412252} use the Inverse Propensity Scoring (IPS) for \textcolor{black}{unbiased} estimation. In this paper, the proposed debiased causal world model explicitly \textcolor{black}{models} the dynamics of popularity instead of the static popularity in existing methods.

\subsection{Debiased Recommendation Systems}
The proposed method is also related to the debiased recommendation system \cite{chen2023bias,liu2023debiased,huang2023debiasing}, which usually borrows the solutions of causal effect \cite{cheng2021long,cai2022rest}. Sharma et al. \cite{sharma2015estimating} estimate the causal effect of \textcolor{black}{the} recommendation system from observed data. Bonner et.al \cite{bonner2018causal} propose the CausE to generate unbiased \textcolor{black}{predictions} under random exposure. Recently, Zhang et.al \cite{zhang2021causal} proposed the popularity-bias deconfounding and adjusting method to remove or adjust the popularity via causal intervention. And Wang et.al \cite{wang2021deconfounded} \textcolor{black}{use} the backdoor adjustment to address the impact of confounders. In this paper, we overcome the constraint of observed confounders and generate \textcolor{black}{unbiased} estimations via VAE \cite{kingma2013auto}. However, these methods seldom consider the effect of latent variables. Hence other researchers investigated the debiased recommendation system with the prior causal generation process with latent variables. Recently, Cai et.al \cite{cai2022rest} consider the exposure strategies as latent variables and reconstruct these latent variables via variational autoencoder. However, without identifying the latent variables, it is hard to guarantee that the model reconstructs the true causal generation process. In this paper, we devise a causal world model and model it with identification guarantees.

\subsection{Identification in Causal Models}
Causal representation learning \cite{scholkopf2021toward,kumar2017variational,locatello2019challenging,locatello2019disentangling,zheng2022identifiability,trauble2021disentangled,li2023subspace} is gaining increasing attention as a means to enhance explanations of generative models. This approach focuses on capturing latent variables of causal generation processes. One prominent contender is independent component analysis (ICA) \cite{hyvarinen2002independent,hyvarinen2013independent,zhang2008minimal,zhang2007kernel,xiemulti,comon1994independent}, a classical avenue for learning causal representations. ICA assumes a linear mixture function for the generation process. Yet, nonlinear ICA poses a formidable challenge due to the non-identifiability of latent variables unless extra assumptions are imposed on either the latent variable distribution or the generation process \cite{hyvarinen1999nonlinear,zheng2022identifiability,hyvarinen2023identifiability,khemakhem2020ice}. Notably, recent work by Aapo et al. \cite{hyvarinen2016unsupervised,hyvarinen2017nonlinear,hyvarinen2019nonlinear,khemakhem2020variational,halva2021disentangling,halva2020hidden} has introduced identification theories that introduce auxiliary variables such as domain indexes, time indexes, and class labels. These approaches typically assume conditional independence of latent variables and adherence to exponential families. In a departure from this exponential families assumption, Zhang et al. \cite{kong2022partial,xiemulti} have extended identification outcomes for nonlinear ICA to include a specific number of auxiliary variables on a component-wise basis. Building upon these theoretical foundations, Yao et al. \cite{yao2022temporally,yao2021learning} have successfully extracted time-delay latent causal variables and their interrelationships from sequential data, even in scenarios with stationary settings and varying distribution shifts. Additionally, Xie et al. \cite{xiemulti} have harnessed nonlinear ICA for reconstructing joint distributions of images from disparate domains. Meanwhile, Kong et al. \cite{kong2022partial} have tackled domain adaptation challenges through component-wise identification results. In this paper, we take time-varying social networks as a special type of supervised signal and reconstruct the latent user preference with an identification guarantee.


\section{Preliminaries}
In this section, we first model the interactive recommendation system as a problem of model-based reinforcement learning. The $t$-th interaction of interactive recommendation can be separated into three steps: 1) the recommendation system \textcolor{black}{suggests} an \textcolor{black}{item} $a_t \in \mathcal{A}$ to a user $u \in \mathcal{U}$\footnote{With the abuse of notation, we let $a$ be the items or its embedding for easy comprehension, and the same as user $u$.}; 2) the user $u$ provides feedback $y_{t}$; 3) the future user \textcolor{black}{preference is} influenced by the feedback of neighbors and the historical user \textcolor{black}{preference}. The goal of the interactive recommendation system is to learn \textcolor{black}{a} model that can suggest as many as acceptable items to users in limited interactions.

\textbf{Online Interactive Recommendation. }Based on the aforementioned process of interaction, we can consider the process of interactive recommendation as a Markov Decision \textcolor{black}{Process} (MDP), which \textcolor{black}{is denoted} by a tuple $\mathcal{M}=(\mathcal{\bm{S}}^u, \mathcal{A}, \rho, R, H, \gamma)$. Here, we let $\mathcal{\bm{S}}^u$ be the latent state space of users, which models the preference of users and \textcolor{black}{is} influenced by the behaviors of neighbors and \textcolor{black}{their} historical \textcolor{black}{preference}; $\mathcal{A}$ denotes the action space which consists of items; $\rho: \mathcal{\bm{S}}^u \times \mathcal{A} \rightarrow \mathcal{\bm{S}}^u$ is the transition function and $R: \mathcal{\bm{S}}^u \times \mathcal{A} \rightarrow y \in \{0, 1\}$ is the reward function; $H$ is the horizon and $\gamma$ is the discount factor to balance the immediate rewards and future rewards. In this way, we model interactive recommendation as a reinforcement learning problem. The object of \textcolor{black}{the} interactive recommendation system is to optimize a policy $\pi:\mathcal{S}^u \rightarrow \mathcal{A}$ to maximize the expected $\gamma$-discounted cumulative reward $\bm{J}_{\rho}(\pi):=\mathbb{E}_{\tau \sim P(\tau|\pi,\rho)}\left[\sum_{t=0}^{H}\gamma_t R_t\right]$, in which $P(\tau|\pi, \rho)$ \textcolor{black}{denotes} the probability of generating a trajectory $\tau:=[w_0, a_0, y_0, \cdots, w_H, a_H, y_H]$ under the policy $\pi$ and the transition function $\rho$, and $w_t=\{z_t, \bm{G}_t\}$ denotes global attributes including popularity and behaviors of neighbors. 

\textbf{Offline Interactive Recommendation. }
\textcolor{black}{Due to the expensive cost, it is unrealistic to directly access the real environment to obtain the rewards. As a result,} the offline reinforcement learning setting is often employed to learn a policy with the historical trajectory $\mathcal{D}=\{\tau^1, \cdots,\tau^N\}$, in which $N$ is the number of \textcolor{black}{users}. In the case of model-based reinforcement learning, we aim to optimize $\bm{J}_{\rho}(\pi)$ by accessing the simulated environment modeled via $\mathcal{D}$.

\section{identifiable Debiased Model-based Interactive Recommendation}
\subsection{Model Overview}\label{sec:problem_def}
The framework of \textcolor{black}{the} proposed \textbf{iDMIR} model is shown in \textcolor{black}{Figure \ref{fig:model}}, including the debiased causal world model and the debiased \textcolor{black}{contrastive} policy. The debiased causal world model, which is developed based on the causal mechanism shown in Figure \ref{fig:model}(a), is used to remove the bad impact of the popularity bias according to the dynamic popularity and imitate \textcolor{black}{unbiased} feedback. The debiased \textcolor{black}{contrastive} policy shown in Figure \ref{fig:model} (c), which is in agreement with the debiased \textcolor{black}{contrastive} learning, is used to avoid the bad impact of sampling bias and further select the optimum items. Hence, \textcolor{black}{as shown in Figure \ref{fig:model}(b)}, the proposed \textbf{iDMIR} works under the framework of model-based reinforcement learning, where the debiased \textcolor{black}{contrastive} policy (\textbf{policy}) suggests \textcolor{black}{unbiased} items (\textbf{actions}) and the debiased causal world model (\textbf{simulated environment}) provides \textcolor{black}{unbiased} feedback (\textbf{rewards}).
\begin{figure*}[t]
		\centering
\includegraphics[width=\columnwidth]{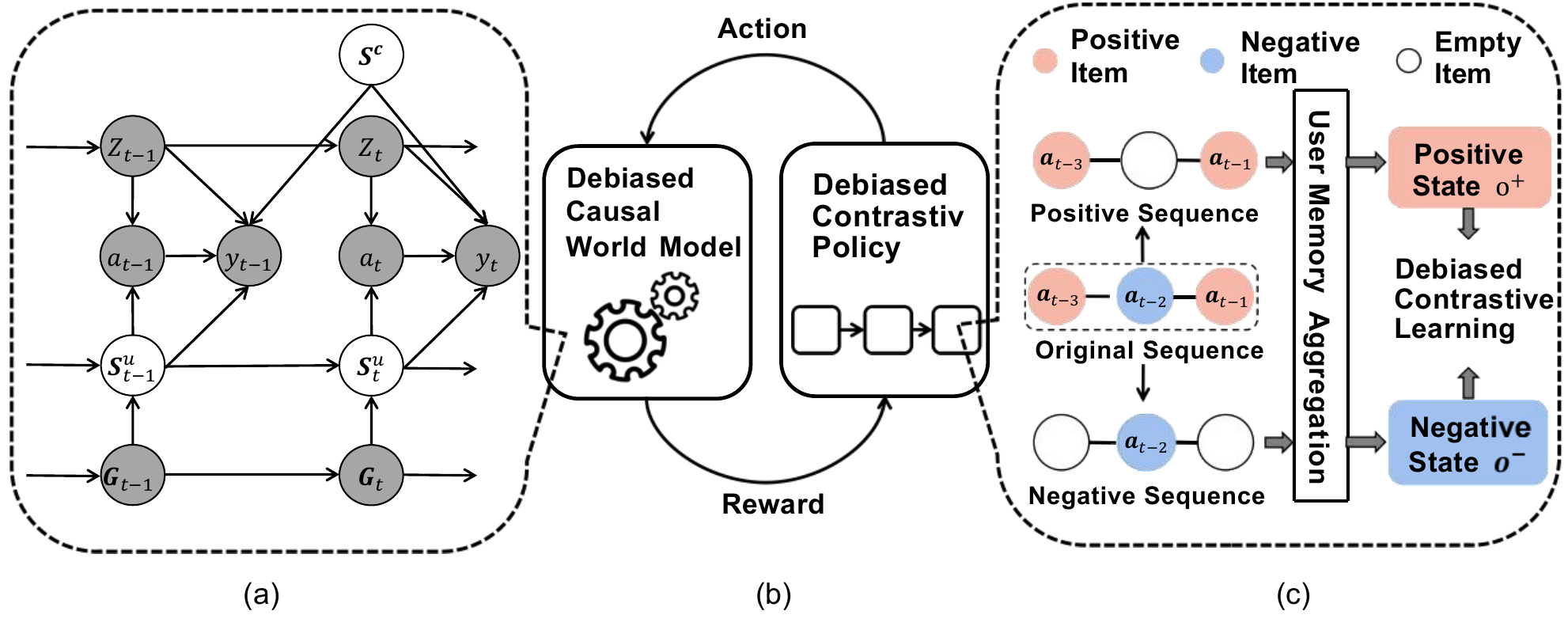}
		\caption{An illustration of iDMIR. (b) \textcolor{black}{illustrates how the debiased contrastive policy and the debiased causal world model are under the framework of model-based reinforcement learning.} (a) Causal graph of debiased causal world model in \textcolor{black}{the} interactive recommendation. (c) The debiased contrastive policy respectively use the positive and negative sequence for unbiased policy. }
		\label{fig:model}
\end{figure*}

\subsection{Learning Debiased Causal World Model with Identification Guarantees}
Since it is too expensive to directly access the real environment, a world model \cite{ha2018world} is desired to answer \textcolor{black}{the} question like ``What would the feedback $y_t$ be if a policy suggested an item $a_t$ under user preference $\bm{s}_t^u$ and context $\bm{s}^c$''. 
\textcolor{black}{To this end, we propose the causal mechanism of the interactive recommendation shown in Figure 2 (a).} First, $\bm{s}_t^u, a_t, z_t, \bm{s}^c \rightarrow y_t$ \textcolor{black}{denotes the reward function $R$ and describes a decision process}, \textcolor{black}{where} the feedback $y_t$ is determined by the user preference $\bm{s}_t^u$, the suggested item $a_t$, the $t$-th popularity $z_t$ and the other time-invariant context information $\bm{s}^c$. Second, $\bm{G}_t, \bm{s}_{t-1}^u \rightarrow \bm{s}_t^u$ \textcolor{black}{\textcolor{black}{denotes} the transition function $\rho$ and describes the user preference transformation process}, where \textcolor{black}{users} \textcolor{black}{preference is} influenced by the preference of neighbors and their historical \textcolor{black}{preference}.Third, $z_{t-1} \!\! \rightarrow \!\! z_t$ and $\bm{G}_{t\!-\!1}\! \! \rightarrow \!\!\bm{G}_t$ \textcolor{black}{denote} the dynamic of popularity and preference of neighbors, respectively.
\textcolor{black}{This causal mechanism explicitly includes the necessary elements of \textcolor{black}{the environment of reinforcement learning} ($R$ and $\rho$), so it is suitable to build the world model.} 
\textcolor{red}{In summary, the data generation process is formalized as Equation (\ref{equ:gen}):}
\begin{equation}
\label{equ:gen}
\color{red}
\begin{split}
    \underbrace{\bm{s}^c \sim p_{c}(\bm{s}^c),}_{\text{context information}} \underbrace{\bm{s}_t^u=\rho(\bm{s}_{t-1}^u,\bm{G}_t)}_{\text{state transition}}, \underbrace{y_t=g(\bm{s}_t^u,a_t,z_t,\bm{s}^c)}_{\text{reward generation}}, \\ \underbrace{z_{t}=f_z(z_{t-1})}_{\text{popularity transition}}, \underbrace{\bm{G}_{t}=f_G(\bm{G} _{t-1})}_{\text{social networks transition}}
\end{split}
\end{equation}
\textcolor{red}{where $\rho$ denotes the transition function and $g$ denotes the function of reward generation.}
Based on the above generation process, answering the aforementioned question equals to \textcolor{black}{estimating} the following \textcolor{red}{intervention} distribution: 
\begin{equation}
\label{equ:do_cal}
   P(y_t|y_{t-1},\bm{G}_t, \bm{G}_{t-1}, z_t, z_{t-1}, a_{t-1}, \bm{s}_{t-1}^u, do(a_t)),
\end{equation}
where $do(\cdot)$ denotes the \textit{do-calculus} \cite{pearl2009causality}. Note that $\bm{s}_{t-1}^u$ can be estimated recursively based on the state transition function which will be introduced in the implementation details of the debiased causal world model.

\textcolor{red}{Several methods \cite{chen2023bias,liu2023debiased,huang2023debiasing} employ the causal generation process for debiased estimation, but these methods usually do not take the latent variables into account, leading to the negative influence of unobserved confounders in real-world scenarios. Recently, Cai et.al. \cite{cai2022rest} involve the exposure strategies as latent variables for the debiased social recommendation. However, they do not provide any identification guarantees for the latent variables, which cannot ensure the theoretical correctness in reconstructing latent variables and further the debiased estimation. To identify the debiased estimation as shown in Equation (\ref{equ:do_cal}), we provide a series of theories as shown in Figure \ref{fig:theory}. Specifically, we use Theorem 1 to make sure that Equation (\ref{equ:do_cal}) can be well reconstructed when the causal generation process can be well learned. Moreover, to guarantee that the causal generation process can be well learned, we propose Theorem 2 and Lemma 3 to make sure that the latent variables $\bm{s}_{t}^u$ and $\bm{s}^c$ are identifiable, respectively. }


\begin{figure*}[t]
		\centering\		\includegraphics[width=0.4\columnwidth]{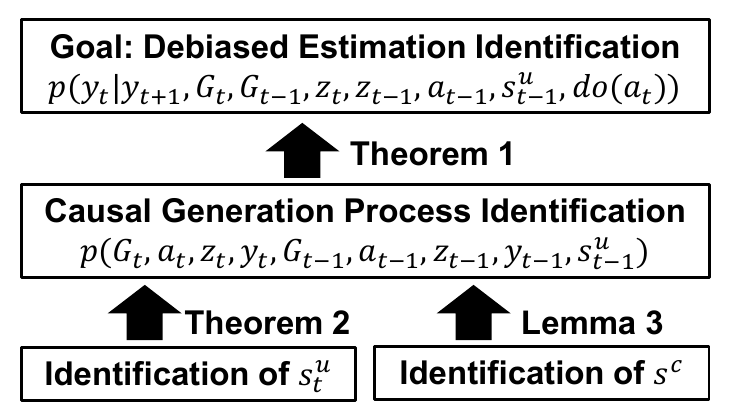}
		\caption{\textcolor{red}{An illustration of theoretical framework for debiased estimation identification.}}
		\label{fig:theory}
\end{figure*}

\subsubsection{\textcolor{red}
{Identification of Debiased Estimation}}

\textcolor{red}{We first show how the \textcolor{red}{intervention} distribution in Equation (\ref{equ:do_cal}) can be identified up to the estimated causal generation process. More formally, the conditional distribution with do-calculus can be inferred from the observation data \cite{peters2017elements}. Therefore, by modeling the joint distribution of observation of two successive timestamps \footnote{\textcolor{red}{In any successive timestamps (i.e., $t-1$ and $t$), we assume that $\bm{s}_{t-1}^u$ is observed since it can be estimated recursively based on the causal mechanism and $\bm{s}_0$ is assumed to be known.}} i.e., $t-1$ and $t$, we can identify the debiased estimation as follows:}


\begin{theorem}
\label{the_id1}
(\textbf{Identification of Debiased Estimation in Interactive Recommendation})\\ Suppose that the joint distribution $P(\bm{G}_t, a_t, z_t, y_t, \bm{G}_{t\!-\!1}, a_{t\!-\!1}, z_{t\!-\!1}, y_{t\!-\!1}, \bm{s}_{t\!-\!1}^u)$ is recovered, then the Equation (\ref{equ:do_cal}) can be estimated under the causal model shown in Figure 2 (a).
\end{theorem}

\textit{Sketch of proof:}(See  \ref{appendix_the1} for full proof). The basic idea is that when the joint distribution is recovered, the \textcolor{red}{intervention} distribution can be estimated by the \textcolor{black}{integral} of latent variables, which is shown as follows:
\begin{equation}
\label{equ:sketch_proof}
\begin{split}
    &P(y_t|y_{t\!-\!1}, \bm{G}_t, \bm{G}_{t\!-\!1}, z_t, z_{t\!-\!1}, a_{t\!-\!1}, \bm{s}_{t\!-\!1}^u, do(a_t))\\=&\int \!\! \underbrace{P(y_t|\bm{s}^c, z_t, a_t, \bm{s}_t^u)}_{(i)} \underbrace{P(\bm{s}_t^u|\bm{s}_{t\!-\!1}^u,\bm{G}_t,a_{t},z_{t})}_{{(ii)}}  \underbrace{P(\bm{s}^c|y_{t\!-\!1}, \bm{s}_{t\!-\!1}^u, z_{t\!-\!1}, a_{t\!-\!1})}_{\text(iii)}d\bm{s}_t^u d \bm{s}^c.
\end{split}
\end{equation}
\textcolor{black}{By doing this, we can leverage terms $(i)\sim(iii)$ to estimate the \textcolor{red}{intervention} distribution shown in Equation (\ref{equ:sketch_proof}). The implementation details will be introduced by recovering the joint distribution. Technologically, we approximate the latent user preference $\bm{s}_t^u$ of \textcolor{black}{terms} $(ii)$ by Equation (\ref{equ:s_enc})\textcolor{black}{,} and approximate the time-invariant context information $\bm{s}^c$ of term ($iii$) by Equation (\ref{equ:beta}). Sequentially, we can estimate unbiased feedback of term ($i$) by Equation (\ref{equ:pred_y}). In summary, the intervention distribution is theoretically computable based on the debiased causal world model.}

Based on the causal mechanism, we model the joint distribution with the help of stochastic variational inference and derive the evidence lower bound (ELBO) as shown in Equal (\ref{equ:elbo}) (See more details in \ref{appendix_elbo}).

\begin{equation}
\label{equ:elbo}
\begin{split}
\mathcal{L}_{ELBO}=&
    -D_{KL}(\hat{P}(\bm{s}_{t}^u|\bm{s}_{t\!-\!1}^u,\bm{G}_t,z_{t},a_{t})||P(\bm{s}_{t}^u|\bm{s}_{t\!-\!1}^u,\bm{G}_t,z_{t},a_{t}))\\
&-D_{KL}(\hat{P}(\bm{s}^c|z_{t\!-\!1},a_{t\!-\!1},y_{t\!-\!1},\bm{s}_{t\!-\!1}^u)||P(\bm{s}^c|z_{t\!-\!1}, a_{t\!-\!1}, y_{t-1}, \bm{s}_{t\!-\!1}^u))\\
&+E_{\hat{P}(\bm{s}_t^u|\bm{s}_{t\!-\!1}^u,\bm{G}_t, a_t, z_t)}E_{\hat{P}(\bm{s}^c|z_{t\!-\!1},a_{t\!-\!1},y_{t\!-\!1},\bm{s}_{t\!-\!1}^u)}\ln P(y_{t}|z_{t},a_{t},\bm{s}_{t}^u,\bm{s}^c)\\
&+E_{\hat{P}(\bm{s}_t^u|\bm{s}_{t\!-\!1}^u,\bm{G}_t, a_t, z_t)}E_{\hat{P}(\bm{s}^c|z_{t\!-\!1},a_{t\!-\!1},y_{t\!-\!1},\bm{s}_{t\!-\!1}^u)}\ln P(y_{t\!-\!1}|z_{t\!-\!1},a_{t\!-\!1},\bm{s}_{t\!-\!1}^u),
\end{split}
\end{equation}
where $D_{KL}(\cdot|\cdot)$ denotes the Kullback-Leibler divergence; $\hat{P}(\bm{s}^c|y_{t-1}, \bm{s}_{t\!-\!1}^u, z_{t\!-\!1}, a_{t\!-\!1})$ and $\hat{P}(\bm{s}_t^u|\bm{s}_{t\!-\!1}^u,\bm{G}_t, a_t, z_t)$ are used to approximate the distribution of $\bm{s}^c$ and $\bm{s}_t^u$; In addition, $P(y_{t-1}|z_{t-1},a_{t-1},\bm{s}_{t-1}^u,\bm{s}^c)$ and $P(y_{t}|z_{t},a_{t},\bm{s}_{t}^u,\bm{s}^c)$ denote the feedback predictors. 

\textcolor{red}{However, simply optimizing the ELBO cannot make sure that the latent variables i.e. $\bm{s}_{t}^u$ and $\bm{s}_{c}$ can be correctly reconstructed with theoretical guarantees. To solve this problem, we further propose Theorem 2 and Lemma 3 to show these latent variables can be reconstructed with theoretical guarantees in the following subsections.}

\subsubsection{\textcolor{red}
{Identification of user preference latent variables $\bm{s}_t^u$}}

\textcolor{red}{In this subsection, we first show that the user preference latent variables can be identified up to component-wise transformations. Formally, for each true latent variable $\bm{s}_{t,i}^u$, there exists a corresponding estimated latent variables $\hat{\bm{s}}_{t,i}^u$ and an invertible function $h_{u,i}:\mathbb{R}\rightarrow \mathbb{R}$, such that $\hat{\bm{s}}_{t,i}^u=h_{u,i}(\bm{s}_{t,i}^u)$. For better understanding, we suppose that the $\bm{s}^c$ and $\bm{s}_{t}^u$ correspond to components in $\bm{s}_t$ with indices $\{1, \cdots, n_c\}$ and $\{n_c+1,\cdots,n\}$, respectively. More specifically, we let $\bm{s}^c=(\bm{s}_{t,i})_{i=1}^{n_c}$ and $\bm{s}_{t}^{u}=(\bm{s}_{t,i})_{i=n_c+1}^{n}$ \footnote{For convenience, we let $\bm{s}_t=[\bm{s}^c;\bm{s}_t^u]$ and the size of $\bm{s}_c$ and $\bm{s}_u$ be $n_c$ and $n_u$, respectively.}. Then we leverage Theorem 2 to show that the user preference latent variables are component-wise identification as follows. }

\begin{theorem}
\color{red}
\label{the1}
We follow the data generation process in Figure \ref{fig:model}(a) and make the following assumptions:
\begin{itemize}[leftmargin=*]
    \item A1 (\underline{Smooth and Positive Density}): The probability density function of latent variables is smooth and positive, i.e., $p(\bm{s}_t|\bm{s}_{t-1}, G_t, a_t, z_t)$ is smooth and $p(\bm{s}_t|\bm{s}_{t-1}, G_t, a_t, z_t)>0$.
    \item A2 (\underline{Conditional Independence}): Conditioned on $G_t$ and $\bm{s}_{t-1}$, $\bm{s}_{t,i}$ is independent of any other $\bm{s}_{t,j}$ for $i,j \in [n], i \neq j$, i.e., $\log p(\bm{s}_t|\bm{s}_{t-1}, G_t, a_t, z_t)=\sum_i^n q_i(\bm{s}_{t,i},\bm{s}_{t-1}, G_t, a_t, z_t)$, where $q_i$ denotes the log density of the conditional distribution, i.e., $q_i:=\log p(\bm{s}_{t, i}|\bm{s}_{t-1}, G_t, a_t, z_t)$
    \item A3 (\underline{Linear Independence}): For any $\bm{s}_t^u\in \mathcal{Z}_t \subseteq \mathbb{R}^{n_u}$, there exist $2n_u+1$ values of $\bm{G}_t$, i.e., $\bm{G}_{t,j}$ with $j=0,1,\cdots,2n_u$, such that the $2n_u$ vectors $\mathrm{w}(\bm{s}_t^u,\bm{s}_{t-1}^u,G_{t,j},a_t,z_t)-\mathrm{w}(\bm{s}_t^u,\bm{s}_{t-1}^u,G_{t,0},a_t,z_t)$ with $j=1,\cdots,2n_u$, are linearly independent, where vector $\mathrm{w}(\bm{s}_t^u,\bm{s}_{t-1}^u,G_{t,j},a_t,z_t)$ is defined as follows:
    \begin{equation}
    \begin{split}
        &\mathrm{w}(\bm{s}_t^u,\bm{s}_{t-1}^u,G_{t,j},a_t,z_t)\\&=\left(\frac{\partial q_{0}(\bm{s}_{t,0}^u,\bm{s}_{t-1}^u,\bm{G}_t,a_t,z_t)}{\partial \bm{s}_{t,0}^u},\cdots, \frac{\partial q_{n-1}(\bm{s}_{t,n-2}^u,\bm{s}_{t-1}^u,\bm{G}_t,a_t,z_t)}{\partial \bm{s}_{t,n-1}^u},\cdots, \right .\\&\left . \frac{\partial^2 q_{0}(\bm{s}_{t,0}^u,\bm{s}_{t-1}^u,\bm{G}_t,a_t,z_t)}{\partial (\bm{s}_{t,0}^u)^2}, \cdots \frac{\partial^2 q_{n-1}(\bm{s}_{t,n-2}\bm{s}_{t,0}^u,\bm{s}_{t-1},\bm{s}_{t,0}^u,\bm{G}_t,a_t,z_t)}{\partial (\bm{s}_{t,n-1}^u)^2} \right).
    \end{split}
    \end{equation}
\end{itemize}
By learning the data generation process, ${\bm{s}}_t^u$ is component-wise identifiable.
\end{theorem}
\textcolor{red}{\textbf{Proof sketch.} More proof is provided in \ref{appendix_the2}.The proof of Theorem \ref{the1} can be separated into three steps. First, we construct an invertible transformation $h$ between the ground-truth user preference latent variables $\bm{s}_{t}^u$ and the estimated ones $\hat{s}_{t}^u$. Second, we leverage the variance of the time-varying social networks to construct a full-rank linear system, which contains only zero solutions. Finally, we show that the user preference latent variables are component-wise identifiable by leveraging the invertibility of the Jacobian of $h$. }

\noindent\textcolor{red}{\textbf{Discussion.} Theorem 2 shows that we can identify the user preference latent variables by using sufficiently varying social networks (i.e., $2n_u+1$ different values of network structures). We can explain this conclusion in a heuristic real-world observation: The more abundant the user behaviors are, the more complex the social networks become, and the easier and more accurately the model can capture the user preference information. }

\subsubsection{\textcolor{red}{Identification of user context latent variables $\bm{s}^c$}}
\textcolor{red}{Based on the component-wise identification of user preference latent variables, we further show that the user context latent variables are block-wise identifiable, meaning that the estimated context latent variables have preserved the information of the ground-truth context latent variables. Formally, we aim to show that there exists an invertible mapping $h'_c: \mathcal{S}^c \rightarrow \mathcal{S}^c$ between the estimated context latent variables and the true part. }

\begin{lemma}
\color{red}
\label{lemma}
We follow the data generation process in Figure \ref{fig:model}(a) and make assumptions A1-A3. Moreover, we make the following assumption. For any set $A_{s}\in \mathcal{S}$ with the following two properties:
\begin{itemize}
    \item  $A_{s}$ has nonzero probability measure, i.e., $\mathbb{P}[\{s\in A_{s}\}|\{\bm{G_t}=\bm{G}'\}]>0$ for any $\bm{G}' \in \mathcal{G}$.
    
    \item $A_{s}$ cannot be expressed as $B_{\bm{s}^c} \times \mathcal{S}^u$ for any $B_{\bm{s}^c} \subset \mathcal{S}^{c}$.
\end{itemize}
$\exists \bm{G}_1,\bm{G}_2 \in \mathcal{G}$, such that 
\begin{equation}
    \int_{s\in A_{s}} P(s|G_1)ds \neq \int_{s\in A_{s}} P(s|G_2)ds .
\end{equation}
By modeling the data generation process, $\bm{s}^c$ is block-wise identifiable.
\end{lemma}
\textcolor{red}{The proof of Lemma 3 can be found in the Appendix \ref{appendix_lemma3}. Lemma 3 shows that the context latent variables can be block-wise identifiable when the $p(s|G)$ changes sufficiently across different timestamps.}
\textcolor{red}{In summary, we can identify the user preference and context latent variables and further identify the intervention probability.}

\subsection{\textcolor{red}{ Implementation of identifiable Debiased Model-based Interactive Recommendation System}}

\textcolor{red}{Based on the theoretical results, we propose the identifiable Debiased Model-based Interactive Recommendation System (iDMIR) as shown in Figure \ref{fig:model}, which contains a debiased causal world model and a debiased contrastive policy.}

\subsubsection{Implementation of $\hat{P}(\bm{s}^u_t|\bm{s}^u_{t\!-\!1},\bm{G}_t, a_t, z_t)$.} 

We can find that $\hat{P}(\bm{s}_t^u|\bm{s}_{t\!-\!1}^u,\bm{G}_t, a_t, z_t)$ contains \textcolor{black}{a} recursive form, reflecting that the current user preference \textcolor{black}{is} controlled by the historical user's and neighbors' preference, the popularity as well as suggested items, \textcolor{black}{so we employ \textcolor{black}{a} recursive neural architecture to reconstruct the latent user preference and let $\bm{s}_0^u$ be a zero vector}. Formally, we have: 
\begin{equation}
\label{equ:s_enc}
    \bm{s}_{t}^u=f_s(\bm{s}^u_{t-1},\bm{G}_t,a_t, z_t;\theta_s),
\end{equation}
in which $\theta_s$ is the training parameter. 
$f_s(\cdot)$ consists of a self-attention network, two feed-forward networks (FFN), a gated recurrent neural network (GRU), two linear layers, and layer normalization with a hidden state size of 64 dimensions. 

Please note that Equation (\ref{equ:s_enc}) can reconstruct $\bm{s}_t^u$ recursively, so it is a special type of recursive neural \textcolor{black}{networks} and can be considered as a cell function.
We further assume that  $P(\bm{s}_t^u|\bm{s}^u_{t\!-\!1},\bm{G}_t, a_t, z_t)$ are delta distributions, so the value of $D_{KL}(\hat{P}(\bm{s}_t^u|\cdot)||P(\bm{s}_{t}^u|\cdot)$ (we ignore the conditional variables for convenience) equals to 0 \textcolor{black}{according to} Proposition \ref{proposition}. \textcolor{black}{The proof} is provided in supplementary materials.

\begin{proposition}
\label{proposition}
(\textbf{KL-Divergence under Delta Distribution Assumption.}) The KL-divergence $D_{KL}(\hat{P}(\bm{s}_t^u|\cdot)||P(\bm{s}_{t}^u|\cdot)$ is zero if $P(\bm{s}_t^u|\cdot)$ is a delta distribution with the optimal parameters $\hat{P}^*=argmax_{\hat{P}}{\mathcal{L}_{ELBO}}$. 
\end{proposition}
\begin{proof}
We are proof by contradiction. First, we suppose that $D_{KL}(Q(\bm{s}_t^u|\cdot)||P(\bm{s}_{t}^u|\cdot)\neq 0$. Then, given the delta distribution $P(\bm{s}_t^u|\cdot)$, there must exist an instance $\bm{s}_t^u$ such that $Q*(\bm{s}_t^u|\cdot)\neq0, P(\bm{s}_t^u|\cdot)=0$. It follows that $D_{KL}(Q(\bm{s}_t^u|\cdot)||P(s_{t}^u|\cdot)\rightarrow \infty$ and results in an under-optimized score $\mathcal{L}_{ELBO}\rightarrow -\infty$, which is a contradiction.
\end{proof}

\subsubsection{Implementation of $\hat{P}(\bm{s}^c|z_{t\!-\!1},a_{t\!-\!1},y_{t\!-\!1},\bm{s}^u_{t\!-\!1})$.} Similar to Equation (\ref{equ:s_enc}), we employ another  to approximate $P(\bm{s}^c|z_{t\!-\!1},a_{t\!-\!1},y_{t\!-\!1},\bm{s}^u_{t\!-\!1})$. Different from $P(\bm{s}^u_t|\cdot)$, we assume $P(\bm{s}^c)\sim \mathcal{N}(0,1)$. So we have:
\begin{equation}
\label{equ:beta}
\begin{split}
\bm{\mu}^c,\bm{\sigma}^c&=f_c(y_{t-1}, \bm{G}_{t-1}, z_{t-1}, a_{t-1};\theta_c) \\
    \bm{s}^c&=\bm{\mu}^c + \delta \bm{\sigma}^c\textcolor{black}{,}
\end{split}
\end{equation}
in which $\theta_c$ is the training parameter; $\bm{\mu}^c$ and $\bm{\sigma}^c$ denote the mean and variance. $\delta$ is a noise variable $\delta\sim\mathcal{N}(0, 1)$.
$f_c(\cdot)$ consists of a self-attention network, an FFN, two linear layers, and layer normalization with a hidden state size of 64 dimensions.

\subsubsection{Implementation of $P(y_{t-1}|z_{t-1},a_{t-1},\bm{s}_{t-1}^u)$ and $P(y_{t}|z_{t},a_{t},\bm{s}_{t-1}^u,\bm{s}^c)$.} 

Since $P(y_{t-1}|z_{t-1}, a_{t-1},\bm{s}_{t-1}^u, \bm{s}^c)$ and $P(y_{t}|z_{t}, a_{t},\bm{s}_{t-1}^u, \bm{s}^c)$ share \textcolor{black}{a} similar form, we also employ the same neural architecture to reconstruct the feedback. Formally, we have:
\begin{equation}
\label{equ:pred_y}
    y_t=f_y(z_t, a_t, \bm{s}_t^u, \bm{s}^c;\theta_y),
\end{equation}
in which $\theta_y$ is the training parameter. When predicting $y_t$, we take $\bm{s}^c$ as input, when predicting $y_{t-1}$, we take a zero vector to replace $\bm{s}^c$. 
In summary, we can learn the debiased causal world model by optimizing Equation (\ref{equ:elbo}). During inference, we first reconstruct the $\bm{s}_t^u$ and $\bm{s}^c$, and then predict the feedback of users given $z_t$ and $a_t$ as shown in Equation (\ref{equ:sketch_proof}).

\subsection{Debiased Contrastive Policy}
\textcolor{black}{It is important to build negative samples for \textcolor{black}{recommendation algorithms}. There are mainly two approaches: using the true negative samples and taking the unknown samples as negative samples. On one hand, we train \textcolor{black}{the model} with the true negative samples, but the limited negative samples have \textcolor{black}{exposure} bias caused by the previous recommendation model, and make it hard to aggregate the ideal state for the policy. On the other hand, the unknown samples are treated as negative ones, which is unreasonable, since we do not know the users' preference \textcolor{black}{for} the unobserved items.}

To address these \textcolor{black}{problems}, we propose a debiased contrastive policy, which not only \textcolor{black}{avoids} the sampling bias but also well \textcolor{black}{aggregates} the information \textcolor{black}{about} what the users like and dislike. In detail, we first split a historical selected item sequence into a positive sequence and a negative sequence. For example, we split $seq=(a_{t-3}^+, a_{t-2}^-, a_{t-1}^+)$ into a positive sequence $seq^+=(a_{t-3}^+, a_\emptyset, a_{t-1}^+)$ and a negative sequence $seq^-=(a_\emptyset, a_{t-2}^-, a_\emptyset)$, where $a^+, a^-, a_\emptyset$ denotes the positive, negative and empty items respectively. After that, we can obtain the \textcolor{black}{contrastive} $o_{t-1}$ with a shared gated recurrent unit (GRU) \cite{chung2014empirical} function shown in Equation ( \ref{equ:policy_state}).
\begin{equation}
\label{equ:policy_state}
    o_{t-1} = o_{t-1}^+-o_{t-1}^-=GRU(seq^+)- GRU(seq^-),
\end{equation}
where $o_{t-1}$ denotes the policy state at $t\!-\!1$ timestamp. And we formalize the Q-function as follows:
\begin{equation}
\label{equ:debiased_q}
    Q(o_{t-1}, a_{t-1};\theta_q)=e^{a_t^\mathsf{T} (o_{t-1})},
\end{equation}
where $\theta_q$ is the training parameters of the state value function. This Q-function consists of a self-attention network, two FFNs, a GRU, four linear layers, and layer normalization with a hidden state size of 64 dimensions. In addition, we follow the paradigm of DQN \cite{mnih2015human} and train the Q-network by minimizing the following loss function:
\begin{equation}
\begin{split}
\mathcal{L}_{\theta_q}=\mathbb{E}_{(\bm{o}_{t}, a_{t}, y_{t}, \bm{o}_{t+1})}[(r_t-Q(o_{t}, a_{t};\theta_q))^2],
    r_t=y_{t}+\gamma \max_aQ(o_{t+1}, a_{t+1};\theta_q),
\end{split}
\end{equation}
in which the $r_t$ is the target value and $\theta_q$ is the training parameter of the Q-network. We can also follow the paradigm of double-DQN \cite{van2016deep} and rewrite the target values.

\textcolor{black}{In order to \textcolor{black}{prove} that the Q function in Equation (\ref{equ:debiased_q}) is \textcolor{black}{unbiased}, we raise Proposition \ref{the_id2} based on the unbiased loss in contrastive learning \cite{chuang2020debiased}. This theoretical result shows that the implementation of the debiased contrastive policy can favor the learning of \textcolor{black}{unbiased} item embedding and further the learning of the debiased policy.}

\begin{definition}
(\textbf{Unbaised Loss in contrastive Learning} \cite{chuang2020debiased}) We let $(a,o^+)$ and $(a,o^-)$ be the similar and dissimilar pairs respectively, and the unbiased loss can be formalized as $\mathcal{L}_u= -\log \frac{e^{a_t^\mathsf{T} o_t^+}}{e^{a_t^\mathsf{T} o_t^+}+ e^{a_t^\mathsf{T} o_t^-}}$.
\end{definition}

\begin{proposition}
\label{the_id2}
(\textbf{\textcolor{black}{Unbiased} Q-function}) Given positive trajectory and negative trajectory, Equation (\ref{equ:debiased_q}) can estimate debiased results with the optimal parameters $Q^*=\arg\min_Q\mathcal{L}_{\theta_q}$.
\begin{proof}
Since $\text{sigmoid}(\cdot)$ is a monotone increasing function, we have $e^{a_t^\mathsf{T} (o_t^+-o_t^-)} \propto \log(\text{sigmoid}(e^{a_t^\mathsf{T} (o_t^+-o_t^-)}))$. Moreover, we have:
\begin{equation}
\label{equ:const}
\begin{split}
    -\log(\text{sigmoid}(e^{a_t^\mathsf{T} (o_t^+-o_t^-)}))=-\log\frac{e^{a_t^\mathsf{T} o_t^+-a_t^\mathsf{T} o_t^-}}{e^{a_t^\mathsf{T} o_t^+-a_t^\mathsf{T} o_t^-}+1}=-\log \frac{e^{a_t^\mathsf{T} o_t^+}}{e^{a_t^\mathsf{T} o_t^+}+ e^{a_t^\mathsf{T} o_t^-}}
\end{split}
    ,
\end{equation}
Hence, minimizing $\mathcal{L}_{\theta_q}$ equals to employ the unbias loss in contrastive learning.
\end{proof}
\end{proposition}

According to proposition \ref{the_id2}, we can find that predicting the feedback by simply subtracting the negative state from the positive state is equivalent to the unbiased contrastive loss. So optimizing $\mathcal{L}_{\theta_q}$ is propitious to learning \textcolor{black}{an} unbiased embedding of $a_t$.

\subsection{Model Summary}
By combining the debiased causal world model and the debiased contrastive policy, we train the whole model under the framework of model-based reinforcement learning as shown in Algorithm \ref{algorithm}. The training procedure can be summarized into four steps. First, we pretrain the debiased causal world model with the logged dataset. 
Second, for each user, the debiased contrastive policy chooses items $a_t$ under the $\epsilon$-greedy policy and \textcolor{black}{interacts} with the debiased causal world model to collect simulated trajectory $\tau^u$. Third, we use the simulated trajectory $\tau^u$ to optimize the state value function. Forth, since it is hard to obtain \textcolor{black}{an} ideal pretrained debiased causal world model, we further finetune it after optimizing the state value function.

\begin{algorithm}[htb]    
	\caption{debiased model-based interactive recommendation ({DMIR}).}
	\label{algorithm}
        \textbf{Initialization:} \\
          \hspace*{1em} set $\leftarrow$  \{$\mathcal{D}, \mathcal{K}_c, \mathcal{K}_q$.\} \\
         \hspace*{1em} \{$k,u,i$\} $\leftarrow 1$ \\
        \textbf{Repeat} \\ 
          \hspace*{1em} $\{\! {w_t}, a_t,{y_t},{w_{t+1}}\!\}$
          $\leftarrow$ $\mathcal{D}$  \\
          \hspace*{1em} optimize $\theta_s,\theta_c,\theta_y$ by minimizing $\mathcal{L}_{ELBO}$.\\
        \textbf{Until} $k$ = $\mathcal{K}_c$\\
        \textbf{While} (not convergence):\\
        \hspace*{1em} Initialize a empty simulated trajectory 
        set $\mathcal{D}_{s}=\emptyset$\;\\
        \hspace*{1em} \textbf{Repeat} \\
         \hspace*{2em} Collect $\tau^u$ for each user by interacting with the debiased causal world model and add it into $\mathcal{D}_{s}$; \\
          \hspace*{1em}\textbf{Until} $u = {N}$\\
           \hspace*{1em}\textbf{Repeat} \\
          \hspace*{2em} $\{\! {w_t}, a_t,{y_t},{w_{t+1}}\!\}$ $\leftarrow$ $\mathcal{D}_{s}$ \\
          \hspace*{2em} optimize $\theta_q$ by minimizing $\mathcal{L}_{\theta_q}$.\\
          \hspace*{1em} \textbf{Until} $i$ = $\mathcal{K}_q$\\
            \hspace*{1em}\textbf{Repeat} \\
            \hspace*{2em} $\{\! {w_t}, a_t,{y_t},{w_{t+1}}\!\}$ $\leftarrow$ $\mathcal{D}_s$ \\
            \hspace*{2em} optimize $\theta_s,\theta_c,\theta_y$ by minimizing ELBO.\\
            \hspace*{1em} \textbf{Until} $i$ = $\mathcal{K}_c$\\
            \textbf{End}

\end{algorithm}

\section{Experiments}

\subsection{Setup}
\subsubsection{Datasets}
In order to evaluate the performance of our method, we conduct experiments on three published datasets (including Ciao, Epinions, and Yelp) with explicit feedback.

\begin{itemize}
\item \textbf{Ciao} is a published dataset for social recommendation. The source cites of Ciao allows users to add friends to their `Circle of Trust', and rate items.

\item\textbf{Epinions} is a benchmark dataset for the social recommendation In Epinions, a user can rate and give comments on items. Besides, a user can also select other users as their trusters. Note that we treat the trust graphs as social networks.

\item\textbf{Yelp} is an online review platform where users review local businesses (e.g., restaurants and shops). The user-item interactions and the social networks are extracted in the same way as Epinions.
\end{itemize}

\subsubsection{Evaluation Setting and Metrics}
To evaluate the performance of recommendation algorithms based on reinforcement learning, one should evaluate the trained policy through the online A/B test. But it is too difficult and expensive to conduct the experiment on the real platform. Hence we follow \cite{chen2019large, dulac2015deep,rohde2018recogym,zou2020pseudo} and simulate the real-world environment with the logging data. 
In detail, we simulate the ground truth item embedding and user \textcolor{black}{preference}. For the ground truth item embedding, we follow \cite{zou2020pseudo} and use a standard rank-H-restricted matrix factorization model \cite{rendle2012bpr} to train the ground truth user embedding $\bm{u}_g$ and the item embedding $\bm{a}_g$, i.e., $P(1|\bm{u}_g, \bm{a}_g)=\text{sigmoid}(\bm{u}_g^\mathsf{T}\bm{a}_g)$. 
Since the users might \textcolor{black}{lose} patience when similar items are repeatedly recommended, the probability that user $\bm{u}_g$ will buy $\bm{a}_g$ can be formalized as $P(1|\bm{u}_g, \bm{a}_g)=\text{sigmoid}(\bm{u}_g^\mathsf{T}\bm{a}_g)*\alpha^{c}$, in which $\alpha$ is the decay rate of interest and $c$ denotes how many times user $\bm{u}$ has been recommended item $\bm{a}_g$.

To measure the performance of recommendation algorithms, we consider the widely-used evaluation metrics like Hit \textcolor{black}{Ratio} (HR) and Normalized Discounted Cumulative Gain (NDCG) within the top-K positions (we report K=20 and K=50 in our experiments). In addition, we use the diversity to measure the richness of recommendation results and the F-measure as a combined metric of diversity and hit \textcolor{black}{ratio}. 

\subsubsection{Implementation Details}
We use the Pytorch framework to implement all the methods and deploy them on NVIDIA A100. The hyperparameters used in all the datasets are shown in Table 1, in which buffer size denotes the size of the experience replay buffer of reinforcement learning; update size denotes the policy will be updated when the amount of data in the experience replay buffer is greater than the update size; $\gamma$ denotes the discount factor to balance the immediate rewards and future rewards; target update denotes the update frequency of the target network in the policy; memory size denotes the number of items that the policy can memorize; epsilon start denotes the value of $\epsilon$ at the beginning of the experiment. We repeat each experiment over 5 random seeds.

\begin{table}[]
\centering
\caption{Hyperparameters setting under three datasets.}
\label{tab:hyper}
\resizebox{0.4\textwidth}{!}{%
\begin{tabular}{c|ccc}
\hline
& Ciao  & Epinions & Yelp  \\ \hline
learning rate & 0.001 & 0.001    & 0.001 \\
batch size    & 1024  & 1024     & 1024  \\
buffer size   & 50000 & 50000    & 50000 \\
update size   & 10000 & 10000    & 10000 \\
$\gamma$        & 0.95  & 0.95     & 0.95  \\
target update & 1000  & 1000     & 1000  \\
droprate      & 0.3   & 0.3      & 0.3   \\
dimension     & 64    & 64       & 64    \\
memory size   & 20    & 20       & 20    \\
epsilon start & 0.3   & 0.5      & 0.7   \\ \hline
\end{tabular}%
}
\end{table}

\subsubsection{Compared Methods}
We compare the proposed debiased model-based interactive recommendation model with model-free methods like \textbf{DQN-r} \cite{zhao2018recommendations}, \textbf{DoubleDQN-r} \cite{van2016deep}, \textbf{GAIL} \cite{GAIL}, \textbf{DC} \cite{xiao2021general} and consider model-based methods including \textbf{SOFA}, \cite{10.1145/3383313.3412252}, \textbf{DEMER} \cite{DEMER} and \textbf{PDQ} \cite{zou2020pseudo}.
We also consider the bandit-based method \textbf{HATCH} \cite{yang2020hierarchical}. 
Apart from this, the recommendation method based on graph neural networks like \textbf{LightGCN} \cite{he2020lightgcn} is considered. We also consider the causality-based methods: Popularity-bias Deconfounding and Adjusting \textbf{(PDA)} \cite{zhang2021causal}, IPS Estimator \textbf{(MF-IPS)}  \cite{schnabel2016recommendations} and doubly robust estimator \textbf{(MRDR-DL)} \cite{10.1145/3404835.3462917}.

\subsection{Experiment Results}
In this subsection, we conduct experiments to answer the following three \textbf{R}esearch \textbf{Q}uestions (\textbf{RQ}): \textbf{RQ1: }How is the performance of the proposed method compared with existing methods? \textbf{RQ2: }How the usage of unknown samples in negative sampling aggravates the impact of bias?  Can the proposed debiased contrastive policy mitigate the harmful effect caused by the negative bias?  \textbf{RQ3: }Can the proposed debiased causal world model be equipped \textcolor{black}{with} other reinforcement learning \textcolor{black}{methods} and promote their \textcolor{black}{performances}?  

\begin{table}[t]
\centering
\caption{Experiment results of Ciao Datasets.}
\label{tab:ciao}
\resizebox{0.8\textwidth}{!}{%
\begin{tabular}{c|cccccc}
\toprule
Methods  & F-measure & Diversity & HR@20  & HR@50  & NDCG@20 & NDCG@50 \\ \midrule
LightGCN & -         & -         & 0.2279 & 0.0909 & 0.3399  & 0.1845  \\
PDA      & -         & -         & 0.2659 & 0.1064 & 0.3928  & 0.2144  \\
MF-IPS   & -         & -         & 0.2635 & 0.1039 & 0.3902  & 0.2123  \\
MRDR-DL  & -         & -         & 0.2675 & 0.1047 & 0.3972  & 0.2174  \\
DGRec    & -        & -    & 0.1973 & 0.2587 & 0.3376  & \textbf{0.4692}
\\
HATCH    & 0.2054    & 0.1438    & 0.3593 & 0.3484 & 0.3912  & 0.3689  \\
SOFA     & 0.1251    & 0.0770    & 0.3340 & 0.1752 & 0.3867  & 0.2437  \\
DC       & 0.1343    & 0.0833    & 0.3458 & 0.1924 & 0.3906  & 0.2543  \\
DQN-r    & 0.1177    & 0.0721    & 0.3208 & 0.1706 & 0.3715  & 0.2352  \\
DDQN-r   & 0.1460    & 0.0914    & 0.3623 & 0.2042 & 0.4025  & 0.2661  \\
PDQ      & 0.2106    & 0.1418    & 0.4091 & 0.2859 & 0.4335  & 0.3307  \\
DEMER    & 0.2341    & 0.1549    & 0.4789 & 0.3531 & 0.4888  & 0.3913  \\
GAIL     & 0.2314    & 0.1545    & 0.4605 & 0.3228 & 0.4763  & 0.3682  \\ 
GSA-M    & 0.3134    & 0.3019    & 0.3594 & 0.1716 & 0.4132  & 0.2342
\\ \hline
DMIR &
  \textbf{0.4076} &
  \textbf{0.3440} &
  \textbf{0.5000} &
  \textbf{0.3722} &
  \textbf{0.5244} &
  0.4183 \\ \hline
\end{tabular}%
}
\end{table}

\begin{table}[]
\centering
\caption{Experiment results of Epinion Dataset.}
\label{tab:epinion}
\resizebox{0.8\textwidth}{!}{%
\begin{tabular}{c|cccccc}
\hline
Methods & F-measure       & Diversity       & HR@20           & HR@50           & NDCG@20         & NDCG@50         \\ \hline
LightGCN & -      & -      & 0.2198 & 0.0873 & 0.3250 & 0.1770 \\
PDA      & -      & -      & 0.2053 & 0.0821 & 0.3032 & 0.1655 \\
MF-IPS   & -      & -      & 0.1966 & 0.0766 & 0.2912 & 0.1579 \\
MRDR-DL  & -      & -      & 0.2224 & 0.0871 & 0.3298 & 0.1786 \\
DGRec    & -      & -      & 0.1916 & 0.2709 & 0.3162 & 0.4828 \\
HATCH    & 0.3511 & 0.2761 & 0.4819 & 0.4825 & 0.4817 & 0.4823 \\
SOFA     & 0.1325 & 0.0854 & 0.2953 & 0.1612 & 0.3530 & 0.2266 \\
DC       & 0.1915 & 0.1321 & 0.3481 & 0.2029 & 0.3867 & 0.2605 \\
DQN-r    & 0.1427 & 0.0955 & 0.2822 & 0.1588 & 0.3390 & 0.2210 \\
DDQN-r   & 0.2347 & 0.1713 & 0.3727 & 0.2416 & 0.4099 & 0.2953 \\
PDQ      & 0.2932 & 0.2028 & 0.5288 & 0.3962 & 0.5329 & 0.4324 \\
DEMER    & 0.2933 & 0.2221 & 0.4316 & 0.3641 & 0.4399 & 0.3865 \\
GAIL     & 0.2938 & 0.2252 & 0.4226 & 0.3677 & 0.4308 & 0.3864 \\ 
GSA-M    & 0.3213 & 0.2916 & 0.3677 & 0.2066 & 0.4342 & 0.2401 \\ \hline
DMIR    & \textbf{0.5176} & \textbf{0.4727} & \textbf{0.5719} & \textbf{0.5203} & \textbf{0.5794} & \textbf{0.5379} \\ \hline
\end{tabular}%
}
\end{table}

\begin{table}[]
\centering
\caption{Experiment results of Yelp Dataset.}
\label{tab:yelp}
\resizebox{0.8\textwidth}{!}{%
\begin{tabular}{c|cccccc}
\hline
Methods & F-measure       & Diversity & HR@20           & HR@50           & NDCG@20         & NDCG@50         \\ \hline
LightGCN & -      & -      & 0.1797 & 0.0712 & 0.2656 & 0.1445 \\
PDA      & -      & -      & 0.2836 & 0.1134 & 0.4188 & 0.2286 \\
MF-IPS   & -      & -      & 0.1984 & 0.0776 & 0.2938 & 0.1595 \\
MRDR-DL  & -      & -      & 0.1791 & 0.0701 & 0.2653 & 0.1436 \\
DGRec    & -    & -    & 0.2351 & 0.3968 & 0.3539 & \textbf{0.6847} \\
HATCH    & 0.4475 & 0.4219 & 0.4763 & 0.4756 & 0.4756 & 0.4748 \\
SOFA     & 0.3421 & 0.2949 & 0.4073 & 0.3200 & 0.4421 & 0.3619 \\
DC       & 0.2530 & 0.2023 & 0.3376 & 0.1983 & 0.3751 & 0.2539 \\
DQN-r    & 0.2487 & 0.1971 & 0.3372 & 0.1948 & 0.3747 & 0.2510 \\
DDQN-r   & 0.3703 & 0.3478 & 0.3991 & 0.3543 & 0.4154 & 0.3755 \\
PDQ      & 0.4337 & 0.3301 & 0.6319 & 0.5856 & 0.6129 & 0.5872 \\
DEMER    & 0.4287 & 0.4006 & 0.4610 & 0.4598 & 0.4634 & 0.4615 \\
GAIL     & 0.4468 & \textbf{0.4233} & 0.4731 & 0.4728 & 0.4731 & 0.4731 \\ 
GSA-M    & 0.3707 & 0.3174 & 0.3826 & 0.2478 & 0.4459 & 0.3078 \\\hline
DMIR    & \textbf{0.4546} & 0.3523    & \textbf{0.6406} & \textbf{0.6057} & \textbf{0.6306} & 0.6091 \\ \hline
\end{tabular}%
}
\end{table}

\subsubsection{Debiased Performance (RQ1)}
In this part, we first investigate the debiased performance of the proposed DMIR method, which is shown in Table 1. According to the experiment results, we can find that: 

\begin{itemize}
    \item 
    \textcolor{black}{DMIR} outperforms the other methods on most of the metrics with a large \textcolor{black}{room for} improvement, which is attributed to both the debiased causal world model and the debiased contrastive policy. 
    \item The model-based methods perform better, since these models simulate the debiased estimation and benefit the policy learning.
    \textcolor{black}{Note that DMIR outperforms GAIL in nearly all metrics, except the diversity, }showing that DMIR can well keep a balance between diversity and accuracy.
    \item Among the static methods, PDA \textcolor{black}{achieves} better results than LightGCN, verifying the efficacy of deconfounding. Other causality-based methods do not perform well, this is because the propensity score is hard to 
    \textcolor{black}{be estimated} and might be influenced by the high variance and latent confounders. 
    \item Note that we do not report the F-measure and the diversity for the static methods since these methods might repeatedly recommend similar items in multi-turn interactions, their performance on diversity can be very low and close to zero.
\end{itemize}

Moreover, we also provide qualitative results as shown in Figure \ref{fig:exp_reward}. According to the experiment results, we can find that {iDMIR} outperforms the other methods on the cumulative reward curve, which is attributed to both the debiased causal world model and the debiased contrastive policy. 


\begin{figure*}
\centering
\subfloat[HR@20]{
\centering
\includegraphics[width=0.34\textwidth]{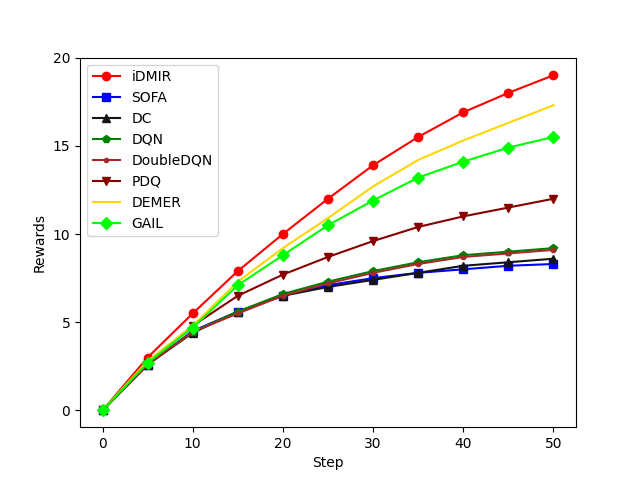}
}
\hspace{-7mm}
\subfloat[HR@50]{
\centering
\includegraphics[width=0.34\textwidth]{{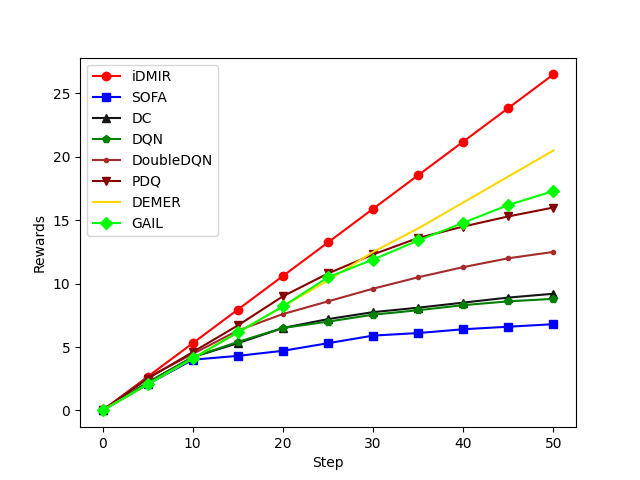}}
}
\hspace{-7mm}
\subfloat[NDCG@20]{
\centering
\includegraphics[width=0.34\textwidth]{{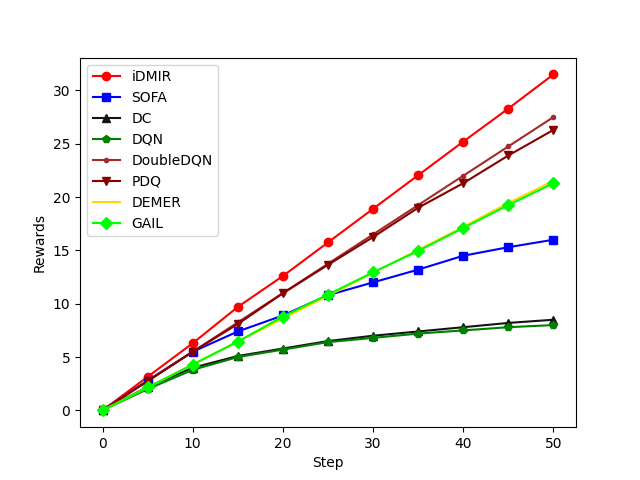}}
}
\hspace{-7mm}
\caption{Evaluation of reward curve results on different datasets.}
\label{fig:exp_reward}
\end{figure*}

\subsubsection{Efficacy of the Debiased contrastive Policy (RQ2)}

In this part, we aim to verify the effectiveness of the debiased contrastive policy. We first answer the question of how the usage of unknown samples in negative sampling aggravates the bad impact of bias. Hence we testify \textcolor{black}{to this} assumption on three reinforcement \textcolor{black}{learning based} methods by keeping and removing the negative sampling. According to Figure 4 on Ciao dataset, most of the reinforcement \textcolor{black}{learning based} methods that remove negative sampling perform better than \textcolor{black}{those} that keep negative sampling, validating that the 
\textcolor{black}{involving} 
unknown samples in negative sampling leads to sampling bias and further degenerates the model performance.

Then we further validate the effectiveness of the debiased contrastive policy. Since both the \textcolor{black}{debiased causal} world model and debiased contrastive policy can mitigate the bad influence of bias, we compare the DMIR-D (the standard DMIR model without debiased causal world model) with the model-free based methods for excluding the effect of world models. 
According to the experiment results on the Ciao and Epinions dataset shown in Figure 3, we find that the debiased contrastive policy outperforms the other model-free based methods, reflecting that the debiased contrastive policy can well aggregate the historical item information. 

\subsubsection{Efficacy of the Counterfactual World Model (RQ3)}


\begin{figure*}
\centering
\subfloat[HR@20]{
\centering
\includegraphics[width=0.23\textwidth]{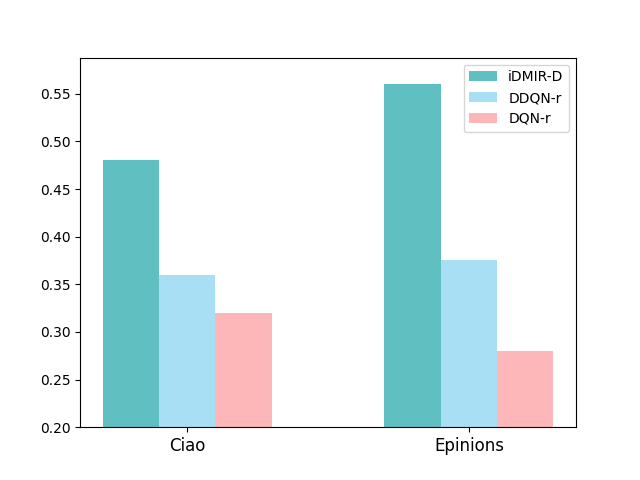}
}\hspace{-4mm}
\subfloat[HR@50]{
\centering
\includegraphics[width=0.23\textwidth]{{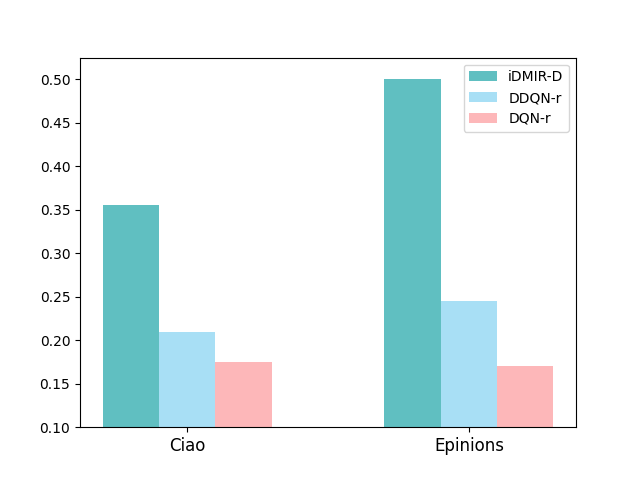}}
}\hspace{-4mm}
\subfloat[NDCG20]{
\centering
\includegraphics[width=0.23\textwidth]{{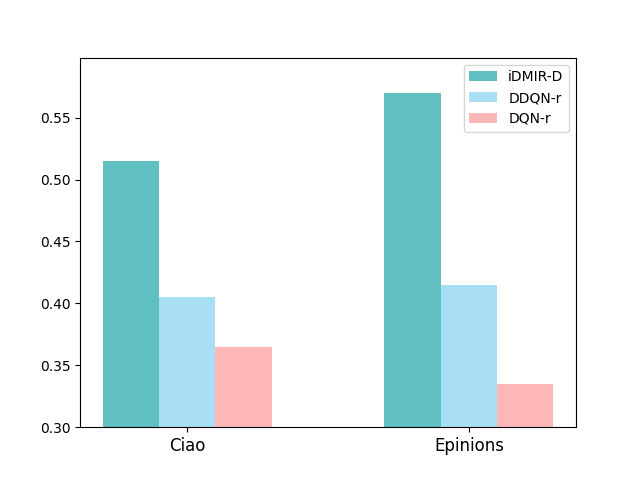}}
}\hspace{-4mm}
\subfloat[NDCG@50]{
\centering
\includegraphics[width=0.23\textwidth]{{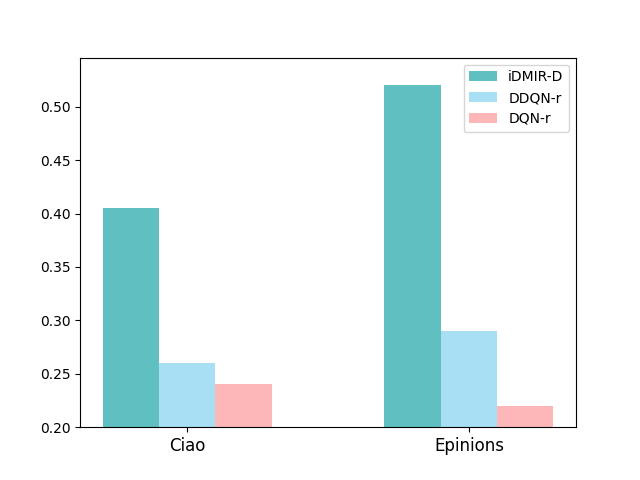}}
}\hspace{-4mm}
\caption{Experiment results among DMIR-D and other model-free based methods}
\end{figure*}



\begin{figure*}[t]
\centering
\subfloat[HR@20]{
\centering
\includegraphics[width=0.23\textwidth]{{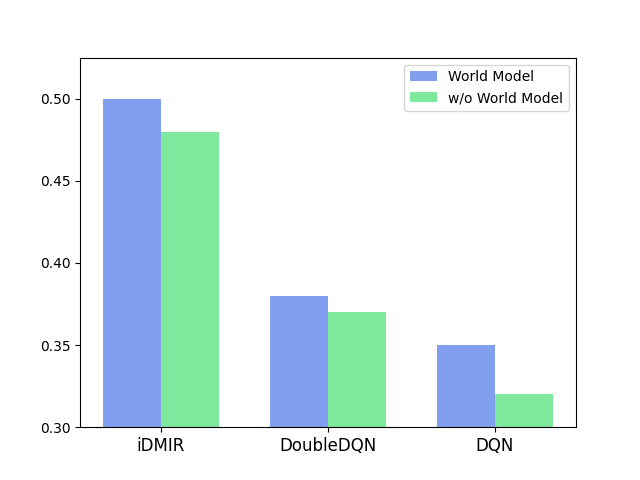}}
}\hspace{-4mm}
\subfloat[HR@50]{
\centering
\includegraphics[width=0.23\textwidth]{{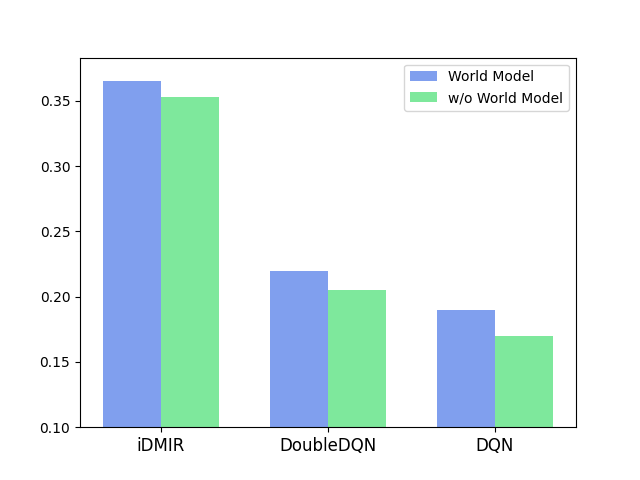}}
}\hspace{-4mm}
\subfloat[NDCG@20]{
\centering
\includegraphics[width=0.23\textwidth]{{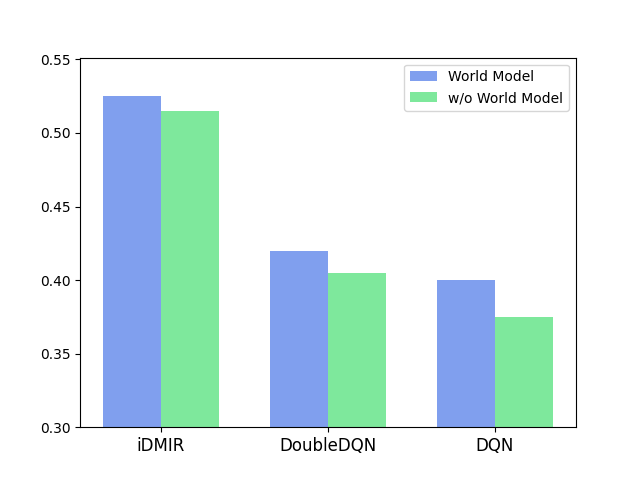}}
}\hspace{-4mm}
\subfloat[NDCG@50]{
\centering
\includegraphics[width=0.23\textwidth]{{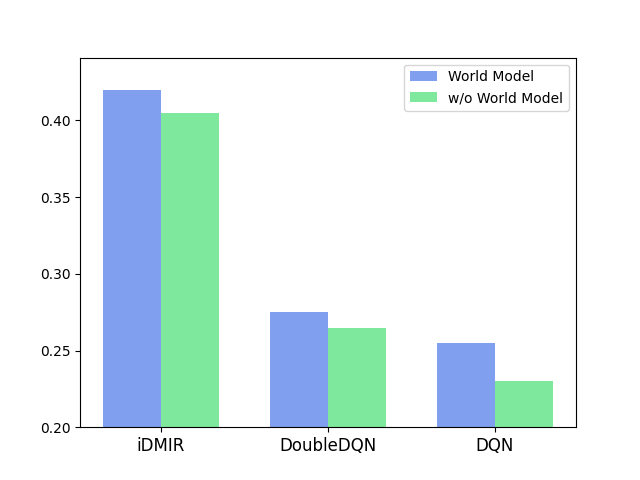}}
}\hspace{-4mm}

\centering
\subfloat[HR@20]{
\centering
\includegraphics[width=0.23\textwidth]{{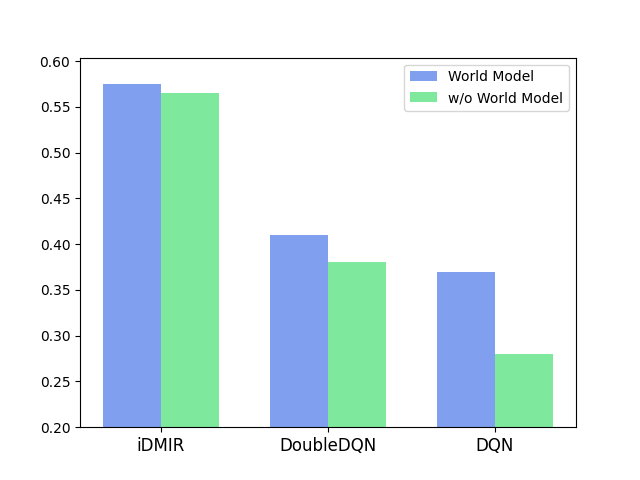}}
}\hspace{-4mm}
\subfloat[HR@50]{
\centering
\includegraphics[width=0.23\textwidth]{{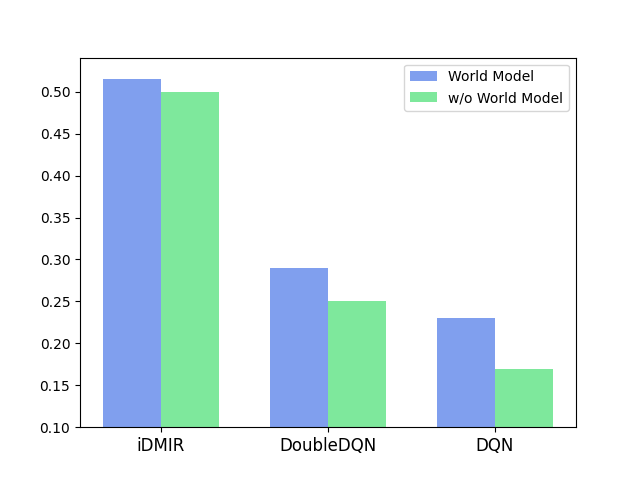}}
}\hspace{-4mm}
\subfloat[NDCG@20]{
\centering
\includegraphics[width=0.23\textwidth]{{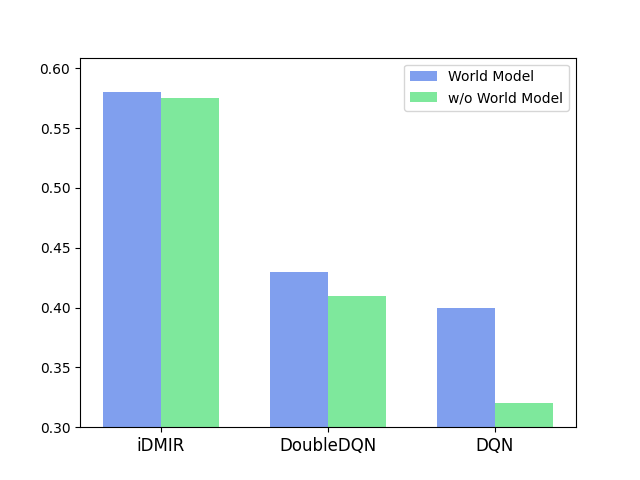}}
}\hspace{-4mm}
\subfloat[NDCG@50]{
\centering
\includegraphics[width=0.23\textwidth]{{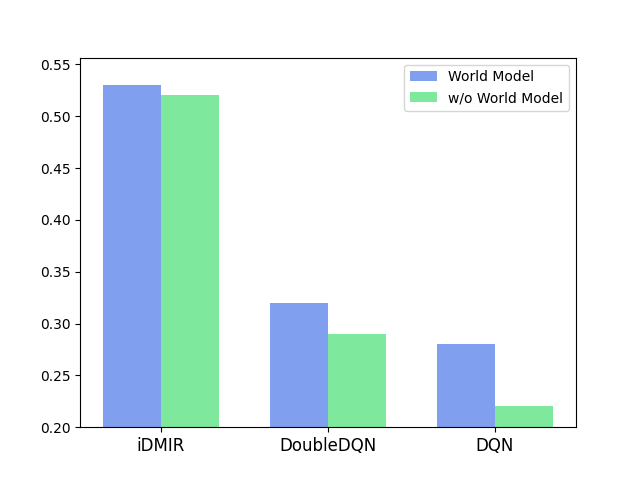}}
}\hspace{-4mm}
\caption{Experiments results between methods that use or remove the debiased causal world model. }
\end{figure*}

To verify the efficacy of the debiased causal world model, we consider it as an off-the-shelf plug-in. Hence we combine the debiased causal world model with the conventional model-free methods like DQN. According to the experiment results shown in Figure 5 (Figure 5 (a)-(d) and Figure 5 (e)-(h) are results of different metrics from Ciao and Epinions), we \textcolor{black}{learn} the following lessons: 1) Compared \textcolor{black}{with} the standard DMIR model and DMIR without world model, the standard DMIR model always performs better than \textcolor{black}{that without} the debiased causal world model. 2) the performance of the model-free based method \textcolor{black}{has been improved }after combining with the debiased causal world model, showing that debiased causal world model can be taken as a flexible plug-in. 

\section{Conclusion}
This paper presents a debiased model-based interactive recommendation model under the offline scenario. For one thing, to \textcolor{black}{simultaneously} consider the dynamics of popularity and remove the bad impact of popularity bias, the proposed method \textcolor{black}{develops} a debiased causal world model based on the causal mechanism. In addition, the identification theory guarantees the feasible \textcolor{black}{unbiased} estimation. For the other thing, to remove the sample bias of negative sampling, we devise the debiased contrastive policy, which coincides with the debiased contrastive learning. The success of the proposed model not only \textcolor{black}{provides} an effective solution for the model-based interactive recommendation problem\textcolor{black}{,} but also \textcolor{black}{provides} \textcolor{black}{an} off-the-shelf plug-in to enhance the performance of model-free methods.

\appendix
\section{Proof of Identification of Debiased Estimation}\label{appendix_the1}
\renewcommand\thetheorem{\arabic{theorem}}
\begin{theorem2}
(\textbf{Identification of Debiased Estimation in Interactive Recommendation}) Suppose that the joint distribution $P(\bm{G}_t, a_t, z_t, y_t, \bm{G}_{t\!-\!1}, a_{t\!-\!1}, z_{t\!-\!1}, y_{t\!-\!1}, \bm{s}_{t\!-\!1}^u)$ is recovered, then the Equation (\ref{equ:do_cal}) can be estimated under the causal model shown in Figure 2 (a).
\end{theorem2}

\begin{proof2}
We prove that  $P(y_t|y_{t\!-\!1}, \bm{G}_t, \bm{G}_{t\!-\!1}, z_t, z_{t\!-\!1}, a_{t\!-\!1}, \bm{s}_{t\!-\!1}^u, do(a_t))$ is identifiable under the premise of the theorem with the help of the following equation:
\begin{equation}
\label{equ:iden1}
\small
\begin{split}
     &P(y_t|y_{t\!-\!1}, \bm{G}_t, \bm{G}_{t\!-\!1}, z_t, z_{t\!-\!1}, a_{t\!-\!1}, \bm{s}_{t\!-\!1}^u, do(a_t))\\
     =&\int P(y_t,\bm{s}_t^u,\beta|C)d\bm{s}_t^u d\bm{s}^c\\
     =&\int \!\!P(y_t|\bm{s}_t^u, \bm{s}^c, C)P(\bm{s}_t^u|\bm{s}^c, C)P
     (\bm{s}^c|C)d\bm{s}_t^u d\bm{s}^c\\
     \overset{(i)}{=}&\int \!\! P(y_t|\bm{s}^c, z_t, a_t, \bm{s}_t^u)P(\bm{s}_t^u|\bm{s}_{t\!-\!1}^u,\bm{G}_t,a_{t},z_{t})P(\bm{s}^c|y_{t\!-\!1}, \bm{s}_{t\!-\!1}^u, z_{t\!-\!1}, a_{t\!-\!1})d\bm{s}_t^u d\bm{s}^c
 \end{split}
\end{equation}
where $C=\{y_{t\!-\!1}, \bm{G}_t, \bm{G}_{t\!-\!1}, z_t, z_{t\!-\!1}, a_{t\!-\!1}, \bm{s}_{t\!-\!1}^u, do(a_t)\}$ and equality $\overset{(i)}{=}$ is by the rule of do-calculus applied to the causal graph shown in Figure 2 (a).
\end{proof2}

\section{Proof of Evidence Lower Bound}
\label{appendix_elbo}
    
\textcolor{black}{The joint distribution $P(\bm{G}_t,a_{t},z_{t},y_{t},\bm{G}_{t\!-\!1},a_{t\!-\!1},z_{t\!-\!1},y_{t\!-\!1},\bm{s}_{t\!-\!1}^u)$ can be modeled by optimizing the ELBO as shown in Equation (\ref{equ:elbo_appendix}).}
    \begin{equation}
        \begin{split}
            &\ln  P(\bm{G}_t,a_{t},z_{t},y_{t},\bm{G}_{t\!-\!1},a_{t\!-\!1},z_{t\!-\!1},y_{t\!-\!1},\bm{s}_{t\!-\!1}^u)\geq ELBO\\
        &ELBO = -D_{KL}(\hat{P}(\bm{s}_{t}^u|\bm{s}_{t\!-\!1}^u,\bm{G}_t,z_{t},a_{t})||P(\bm{s}_{t}^u|\bm{s}_{t\!-\!1}^u,\bm{G}_t,z_{t},a_{t}))\\
    &\quad-D_{KL}(\hat{P}(\bm{s}^c|z_{t\!-\!1},a_{t\!-\!1},y_{t\!-\!1},\bm{s}_{t\!-\!1}^u)||P(\bm{s}^c|z_{t\!-\!1},a_{t\!-\!1},y_{t\!-\!1},\bm{s}_{t\!-\!1}^u))\\
    &\quad+E_{\hat{P}(\bm{s}_t^u|\bm{s}_{t\!-\!1}^u,\bm{G}_t, a_t, z_t)}E_{\hat{P}(\bm{s}^c|z_{t\!-\!1},a_{t\!-\!1},y_{t\!-\!1},\bm{s}_{t\!-\!1}^u)}\ln P(y_{t}|z_{t},a_{t},\bm{s}_{t}^u,\bm{s}^c)\\
    &\quad+E_{\hat{P}(\bm{s}_t^u|\bm{s}_{t\!-\!1}^u,\bm{G}_t, a_t, z_t)}E_{\hat{P}(\bm{s}^c|z_{t\!-\!1},a_{t\!-\!1},y_{t\!-\!1},\bm{s}_{t\!-\!1}^u)}\ln P(y_{t\!-\!1}|z_{t\!-\!1},a_{t\!-\!1},\bm{s}_{t\!-\!1}^u)\ \
        \end{split}
        \label{equ:elbo_appendix}
    \end{equation}

\begin{proof}
The proof of the ELBO is composed of three steps. First, we factorize the conditional distribution according to the Bayes theorem.
    \begin{equation}
        \begin{split}
            &\ln  P(\bm{G}_t,a_{t},z_{t},y_{t},\bm{G}_{t\!-\!1},a_{t\!-\!1},z_{t\!-\!1},y_{t\!-\!1},\bm{s}_{t\!-\!1}^u)\\
            =&\ln \frac{P(\bm{G}_t,a_{t},z_{t},y_{t},\bm{G}_{t\!-\!1},a_{t\!-\!1},z_{t\!-\!1},y_{t\!-\!1},\bm{s}_t^u,\bm{s}_{t\!-\!1}^u,\bm{s}^c)}{P(\bm{s}_t^u,\bm{s}^c|\bm{G}_t,a_{t},z_{t},y_{          t},\bm{G}_{t\!-\!1},a_{t\!-\!1},z_{t\!-\!1},y_{t\!-\!1},\bm{s}_{t\!-\!1}^u)}\\
            =&\ln \frac{P(\bm{G}_t,a_{t},z_{t},y_{t},\bm{G}_{t\!-\!1},a_{t\!-\!1},z_{t\!-\!1},y_{t\!-\!1},\bm{s}_t^u,\bm{s}_{t\!-\!1}^u,\bm{s}^c)}{P(\bm{s}^c|\bm{s}_t^u,a_{t},z_{t},y_{t},\bm{s}_{t\!-\!1}^u,a
            _{t\!-\!1},z_{t\!-\!1},y_{t\!-\!1})P(\bm{s}_t^u|\bm{s}_{t\!-\!1}^u,\bm{G}_t,a_{t},z_{t},y_{t})}
            \\
        \end{split}\nonumber
    \end{equation}
Second, we add the expectation operator on both sides of the equation and reformalize the equation as follows:
    \begin{equation}
        \begin{split}
            &\ln  P(\bm{G}_t,a_{t},z_{t},y_{t},\bm{G}_{t\!-\!1},a_{t\!-\!1},z_{t\!-\!1},y_{t\!-\!1},\bm{s}_{t\!-\!1}^u)\\
            =&D_{KL}(\hat{P}(\bm{s}_t^u|\bm{s}_{t\!-\!1}^u,\bm{G}_t, a_t, z_t)||P(\bm{s}_t^u|\bm{s}_{t\!-\!1}^u,\bm{G}_t,a_{t},z_{t},y_{t}))\\
            &+D_{KL}(\hat{P}(\bm{s}^c|z_{t\!-\!1},a_{t\!-\!1},y_{t\!-\!1},\bm{s}_{t\!-\!1}^u)||P(\bm{s}^c|\bm{s}_t^u,a_{t},z_{t},y_{t},\bm{s}_{t\!-\!1}^u,a_{t\!-\!1},z_{t\!-\!1},y_{t\!-\!1}))\\
            &+\ln \frac{P(\bm{G}_t,a_{t},z_{t},y_{t},\bm{G}_{t\!-\!1},a_{t\!-\!1},z_{t\!-\!1},y_{t\!-\!1},\bm{s}_t^u,\bm{s}_{t\!-\!1}^u,\bm{s}^c)}{\hat{P}(\bm{s}_t^u|\bm{s}_{t\!-\!1}^u,\bm{G}_t, a_t, z_t)
            \hat{P}(\bm{s}^c|z_{t\!-\!1},a_{t\!-\!1},y_{t\!-\!1},\bm{s}_{t\!-\!1}^u)}\\
        \end{split}\nonumber
    \end{equation}
Third, we obtain the last equality with the help of $D_{KL}(\cdot ||\cdot) \geq 0$
    \begin{equation}
        \begin{split}\nonumber
            &\ln  P(\bm{G}_t,a_{t},z_{t},y_{t},\bm{G}_{t\!-\!1},a_{t\!-\!1},z_{t\!-\!1},y_{t\!-\!1},\bm{s}_{t\!-\!1}^u)\\
            \geq& \ln \frac{P(\bm{G}_t,a_{t},z_{t},y_{t},\bm{G}_{t\!-\!1},a_{t\!-\!1},z_{t\!-\!1},y_{t\!-\!1},\bm{s}_t^u,\bm{s}_{t\!-\!1}^u,\bm{s}^c)}{\hat{P}(\bm{s}_t^u|\bm{s}_{t\!-\!1}^u,\bm{G}_t, a_t, z_t)
        \hat{P}(\bm{s}^c|z_{t\!-\!1},a_{t\!-\!1},y_{t\!-\!1},\bm{s}_{t\!-\!1}^u)}\\
        =&\ln \frac {P(y_{t}|z_{t},a_{t},\bm{s}_{t}^u,\bm{s}^c)P(\bm{G}_t,a_{t},z_{t},\bm{G}_{t\!-\!1},a_{t\!-\!1},z_{t\!-\!1},y_{t\!-\!1},\bm{s}_t^u,\bm{s}_{t\!-\!1}^u,\bm{s}^c)}{\hat{P}(\bm{s}_t^u|\bm{s}_{t\!-\!1}^u,\bm{G}_t, a_t, z_t)
        \hat{P}(\bm{s}^c|z_{t\!-\!1},a_{t\!-\!1},y_{t\!-\!1},\bm{s}_{t\!-\!1}^u)}\\
        =&\ln \frac {P(y_{t}|z_{t},a_{t},\bm{s}_{t}^u,\bm{s}^c)
            P(\bm{s}_t^u|\bm{s}_{t\!-\!1}^u,\bm{G}_t,z_t,a_t)
        P(\bm{G}_t,a_t,z_t,\bm{G}_{t\!-\!1},a_{t\!-\!1},z_{t\!-\!1},y_{t\!-\!1},\bm{s}_{t\!-\!1}^u,\bm{s}^c)}
            {\hat{P}(\bm{s}_t^u|\bm{s}_{t\!-\!1}^u,\bm{G}_t, a_t, z_t)
        \hat{P}(\bm{s}^c|z_{t\!-\!1},a_{t\!-\!1},y_{t\!-\!1},\bm{s}_{t\!-\!1}^u))}
        \\
        =& -D_{KL}(\hat{P}(\bm{s}_{t}^u|\bm{s}_{t\!-\!1}^u,\bm{G}_t,z_{t},a_{t})||P(\bm{s}_{t}^u|\bm{s}_{t\!-\!1}^u,\bm{G}_t,z_{t},a_{t}))\\
        & -D_{KL}(\hat{P}(\bm{s}^c|z_{t\!-\!1},a_{t\!-\!1},y_{t\!-\!1},\bm{s}_{t\!-\!1}^u)||P(\bm{s}^c|z_{t\!-\!1},a_{t\!-\!1},y_{t\!-\!1},\bm{s}_{t\!-\!1}^u))\\
        & +E_{\hat{P}(\bm{s}_t^u|\bm{s}_{t\!-\!1}^u,\bm{G}_t, a_t, z_t)}E_{\hat{P}(\bm{s}^c|z_{t\!-\!1},a_{t\!-\!1},y_{t\!-\!1},\bm{s}_{t\!-\!1}^u)}\ln P(y_{t}|z_{t},a_{t},\bm{s}_{t}^u,\bm{s}^c)\\
        & +E_{\hat{P}(\bm{s}_t^u|\bm{s}_{t\!-\!1}^u,\bm{G}_t, a_t, z_t)}E_{\hat{P}(\bm{s}^c|z_{t\!-\!1},a_{t\!-\!1},y_{t\!-\!1},\bm{s}_{t\!-\!1}^u)}\ln P(y_{t\!-\!1}|z_{t\!-\!1},a_{t\!-\!1},\bm{s}_{t\!-\!1}^u)\\
        & +\ln P(\bm{G}_t,a_t,z_t,\bm{G}_{t\!-\!1},a_{t\!-\!1},z_{t\!-\!1},\bm{s}_{t\!-\!1}^u).
        \end{split}
    \end{equation}
    Since $\ln P(\bm{G}_t,a_t,z_t,\bm{G}_{t\!-\!1},a_{t\!-\!1},z_{t\!-\!1},\bm{s}_{t\!-\!1}^u)$ is the joint distribution of observed variables which can be considered as a constant, we remove this term when optimizing the ELBO. Therefore, assuming $\bm{s}_0^u$ is observed and $\bm{s}_{t-1}^u$ can be calculated recursively, we finish deriving the evidence lower bound.
\end{proof}

\section{Proof of Theorem \ref{the1}}
\label{appendix_the2}
\begin{theorem2}
\label{the1}
We follow the data generation process in Figure \ref{fig:model}(a) and make the following assumptions:
\begin{itemize}
    \item A1 (\underline{Smooth and Positive Density}): The probability density function of latent variables is smooth and positive, i.e., $p(\bm{s}_t|\bm{s}_{t-1}, G_t, a_t, z_t)$ is smooth and $p(\bm{s}_t|\bm{s}_{t-1}, \bm{G}_t, a_t, z_t)>0$.
    \item A2 (\underline{Conditional Independence}): Conditioned on $\bm{G}_t$ and $\bm{s}_{t-1}$, $\bm{s}_{t,i}$ is independent of any other $\bm{s}_{t,j}$ for $i,j \in [n], i \neq j$, i.e., $\log p(\bm{s}_t|\bm{s}_{t-1}, \bm{G}_t, a_t, z_t)=\sum_i^n q_i(\bm{s}_{t,i},\bm{s}_{t-1}, G_t, a_t, z_t)$, where $q_i$ denotes the log density of the conditional distribution, i.e., $q_i:=\log p(\bm{s}_{t, i}|\bm{s}_{t-1}, \bm{G}_t, a_t, z_t)$
    \item A3 (\underline{Linear Independence}): For any $\bm{s}_t\in \mathcal{Z}_t \subseteq \mathbb{R}^n$, there exist $2n+1$ values of $\bm{G}_t$, i.e., $\bm{G}_{t,j}$ with $j=0,1,\cdots,2n$, such that the $2n$ vectors $\mathrm{w}(\bm{s}_t,\bm{s}_{t-1},G_{t,j},a_t,z_t)-\mathrm{w}(\bm{s}_t,\bm{s}_{t-1},G_{t,0},a_t,z_t)$ with $j=1,\cdots,2n$, are linearly independent, where vector $\mathrm{w}(\bm{s}_t,\bm{s}_{t-1},G_{t,j},a_t,z_t)$ is defined as follows:
    \begin{equation}
    \begin{split}
        \mathrm{w}(\bm{s}_t,\bm{s}_{t-1},G_{t,j},a_t,z_t)=\left(\frac{\partial q_{0}(\bm{s}_{t,0},\bm{s}_{t-1},\bm{G}_t,a_t,z_t)}{\partial \bm{s}_{t,0}},\cdots, \frac{\partial q_{n-1}(\bm{s}_{t,n-2},\bm{s}_{t-1},\bm{G}_t,a_t,z_t)}{\partial \bm{s}_{t,n-1}},\cdots, \right .\\\left . \frac{\partial^2 q_{0}(\bm{s}_{t,0},\bm{s}_{t-1},\bm{G}_t,a_t,z_t)}{\partial \bm{s}_{t,0}^2}, \cdots \frac{\partial^2 q_{n-1}(\bm{s}_{t,n-2},\bm{s}_{t-1},\bm{G}_t,a_t,z_t)}{\partial \bm{s}_{t,n-1}^2} \right).
    \end{split}
    \end{equation}
\end{itemize}
By learning the data generation process, ${\bm{s}}_t$ is component-wise identifiable.
\end{theorem2}
\begin{proof}
Our proof of component-wise identifiability starts from deriving relations on estimated latent space from observational equivalence, i.e., the joint distribution of the observed variables $p_{\hat{g},\hat{\rho},{\hat{p}}_{\beta}}(a_t, z_t, y_t, \bm{G}_t, a_{t-1}, z_{t-1}, y_{t-1}, \bm{G}_{t-1})$ matches $p_{{g},{\rho},{{p}}_{\beta}}(a_t, z_t, y_t, \bm{G}_t, a_{t-1}, z_{t-1}, y_{t-1}, \bm{G}_{t-1})$ everywhere. 
Let $g$ and $\hat{g}$ be the true and estimated reward generation function, respectively. Then we can rewrite the reward generation function $y_t=(g \circ g^{-1} \circ \hat{g})(\bm{s}^u_t,a_t,z_t,\bm{s}^c)$ because of injective properties of $g,\hat{g}$, we can see that $\hat{g}=g\circ \left(\left(g\right)^{-1}\circ \hat{g}\right)=g \circ h$ for some function $h=(g)^{-1} \circ \hat{g}$ on the latent space. 
According to the data generation process in Figure \ref{fig:model}, we can develop the relationship between $s$ and $\hat{s}$ with the help of the change of variables formula as follows:
\begin{equation}
\label{equ:the1_0}
    p(\bm{s}_t^u,\bm{s}_{t-1}^u, \bm{G}_t)=p(h^{-1}(\hat{\bm{s}}_t^u),h^{-1}(\hat{\bm{s}}_{t-1}^u), \bm{G}_t)|\text{det}\frac{\partial h^{-1}(\hat{\bm{s}}_t^u)}{\partial \bm{s}_t^u}||\text{det}\frac{\partial h^{-1}(\hat{\bm{s}}_{t-1}^u)}{\partial \bm{s}_{t-1}^u}|,
\end{equation}
\begin{equation}
\label{equ:the1_1}
    p(\bm{s}_t^u)=p(h^{-1}(\hat{\bm{s}}_t^u))|det\frac{\partial h^{-1}(\hat{\bm{s}}_t^u)}{\partial \bm{s}_t^u}|,
\end{equation}
\begin{equation}
\label{equ:the1_3}
    p(\bm{s}_{t-1}^u, \bm{G}_t)=p(h^{-1}(\hat{\bm{s}}_{t-1}^u))|det\frac{\partial h^{-1}(\hat{\bm{s}}_{t-1}^u)}{\partial \bm{s}_{t-1}^u} |.
\end{equation}
Solving for the determinant terms in Equation (\ref{equ:the1_1}) and Equation (\ref{equ:the1_3}) and plugging them into Equation (\ref{equ:the1_0}), we have:
\begin{equation}
    p(\bm{s}_t^u|\bm{s}_{t-1}^u,\bm{G}_t)=p(h^{-1}(\hat{\bm{s}}_t^u)|h^{-1}(\hat{\bm{s}}_{t-1}^u), G_t)\frac{p(\bm{s}_t^u)}{p(h^{-1}(\bm{s}_t^u))}.
\end{equation}
Sequentially, we define $q(\bm{s}_t^u)=\log p(\bm{s}_t^u|\bm{s}_{t-1}^u,\bm{G}_t)$ as the marginal log-density of the components $\bm{s}_t^u$, and $\overline{q}(\bm{s}_t^u)=\log p(\bm{s}_t^u)$ We then have:
\begin{equation}
    q(\bm{s}_t^u|\bm{s}_{t-1}^u,\bm{G}_t)-q(h^{-1}(\hat{\bm{s}}_t^u)|h^{-1}(\hat{\bm{s}}_{t-1}^u),\bm{G}_t)=\overline{q}(\bm{s}_t^u) - \overline{q}(h^{-1}(\bm{s}_t^u)).
\end{equation}

Based on the conditional independence assumption, we can further have:
\begin{equation}
\label{equ:the_1-4}
    \sum_i \left(q_i(\bm{s}_t^u|\bm{s}_{t-1}^u,\bm{G}_t)-q_i(h^{-1}(\hat{\bm{s}}_t^u)|h^{-1}(\hat{\bm{s}}_{t-1}^u),\bm{G}_t)\right)=\overline{q}(\bm{s}_t^u) - \overline{q}(h^{-1}(\hat{\bm{s}}_t^u)).
\end{equation}

For the ease of exposition, we adopt the following notation:
\begin{equation}
    q_i^1(\bm{s}_t^u|\bm{s}_{t-1}^u,\bm{G}_t)=\frac{\partial q_i(s_{t,i}^u|\bm{s}_{t-1}^u,\bm{G}_t)}{\partial s_{t,i}^u}, \quad 
    q_i^2(\bm{s}_t^u|\bm{s}_{t-1}^u,\bm{G}_t)=\frac{\partial^2 q_i(s_{t,i}^u|\bm{s}_{t-1}^u,\bm{G}_t)}{\partial (s_{t,i}^u)^2},
\end{equation}
and take derivatives of both sides of the Equation (\ref{equ:the_1-4}) with respective $\bm{s}_{t,j}^u$ and we have:
\begin{equation}
    q_j^1(\bm{s}_{t,j}^u|\bm{s}_{t-1}^u,\bm{G}_t)-\sum_{i=n_c+1}^n q_i^1(h^{-1}_i (\bm{s}_t^u)|h^{-1}(\bm{s}_{t-1}^u),\bm{G}_t)\frac{\partial h_i^{-1}(\bm{s}_t^u)}{\partial \bm{s}_{t, j}^u}=\frac{\partial \overline{q}_{\bm{s}_t^u}}{\partial \bm{s}_{t,j}^u} - \frac{\partial \overline{q}(h^{-1}(\hat{\bm{s}}_t^u))}{\partial \bm{s}_{t,j}^u} 
\end{equation}
Then we take another derivative with respect to $\bm{s}_{t,k}^u$, where $j \neq k$, and we have:
\begin{equation}
\begin{split}
    \sum_{i=n_c+1}^n \left( q_i^2(h^{-1}_i (\bm{s}_t^u)|h^{-1}(\bm{s}_{t-1}^u),\bm{G}_t)\frac{\partial h_i^{-1}(\bm{s}_t^u)}{\partial \bm{s}_{t, j}^u} \frac{\partial h_i^{-1}(\bm{s}_t^u)}{\partial \bm{s}_{t,k}^u} + q_i^1(h^{-1}_i (\bm{s}_t^u)|h^{-1}(\bm{s}_{t-1}^u),\bm{G}_t)\frac{\partial^2 h_i^{-1}(\bm{s}_t^u)}{\partial \bm{s}_{t, j}^u \partial s_{t, k}^u}  \right) =\\ \frac{\partial^2 \overline{q}_{\bm{s}_t^u}}{\partial \bm{s}_{t,j}^u \partial \bm{s}_{t,k}^u} - \frac{\partial^2 \overline{q}(h^{-1}(\hat{\bm{s}}_t^u))}{\partial \bm{s}_{t,j}^u \partial \bm{s}_{t,k}^u},
\end{split}
\end{equation}
where $\frac{\partial^2 \overline{q}_{\bm{s}_t^u}}{\partial \bm{s}_{t,j}^u \partial \bm{s}_{t,k}^u} - \frac{\partial^2 \overline{q}(h^{-1}(\hat{s}_t^u))}{\partial \bm{s}_{t,j}^u \partial \bm{s}_{t,k}^u}$ does not depend on $\bm{G}_t$. Therefore, for $\bm{G}_t=\bm{G}_t^0, \cdots, \bm{G}_t^{2n_u}$, we have $2n_u+1$ such equations. Subtracting each equation corresponding to $G_1,\cdots,G_{2n_u}$ with the equation corresponding to $G_0$ results in $2n_u$ equations:
\begin{equation}
    \begin{split}
    \sum_{i=n_c+1}^n \left( \left(q_i^2(h^{-1}_i (\bm{s}_t^u)|h^{-1}(\bm{s}_{t-1}^u),\bm{G}_t^j) - q_i^2(h^{-1}_i (\bm{s}_t^u)|h^{-1}(\bm{s}_{t-1}^u),\bm{G}_t^0)\right)\frac{\partial h_i^{-1}(\bm{s}_t^u)}{\partial \bm{s}_{t, j}^u} \frac{\partial h_i^{-1}(\bm{s}_t^u)}{\partial \bm{s}_{t,k}^u}\right. \\+ \left.\left(q_i^1(h^{-1}_i (\bm{s}_t^u)|h^{-1}(\bm{s}_{t-1}^u),\bm{G}_t^j)-q_i^1(h^{-1}_i (\bm{s}_t^u)|h^{-1}(\bm{s}_{t-1}^u),\bm{G}_t^0)\right)\frac{\partial^2 h_i^{-1}(\bm{s}_t^u)}{\partial \bm{s}_{t, j}^u \partial s_{t, k}^u}  \right) =0.
\end{split}
\end{equation}
Under the linear independence condition in A3, the linear system is a $2n_u \times 2n_u$ full-rank system. Therefore, $\frac{\partial h_i^{-1}(\bm{s}_t^u)}{\partial \bm{s}_{t, j}^u} \frac{\partial h_i^{-1}(\bm{s}_t^u)}{\partial \bm{s}_{t,k}^u}=0$ and $\frac{\partial^2 h_i^{-1}(\bm{s}_t^u)}{\partial \bm{s}_{t, j}^u \partial s_{t, k}^u}=0$ are the only solutions. Note that $\frac{\partial h_i^{-1}(\bm{s}_t^u)}{\partial \bm{s}_{t, j}^u} \frac{\partial h_i^{-1}(\bm{s}_t^u)}{\partial \bm{s}_{t,k}^u}=0$ means that there are at most one non-zero entry in each row in the Jacobian matrix $J_h$. Since $h$ is invertible, there is only one non-zero entry in each row in the Jacobian matrix $J_h$, meaning that $s_u^t$ is component-wise identification.
\end{proof}

\section{Proof of Lemma \ref{lemma}}
\label{appendix_lemma3}
\begin{lemma2} \cite{kong2022partial}
We follow the data generation process in Figure \ref{fig:model}(a) and make assumptions A1-A3. Moreover, we make the following assumption. For any set $A_{s}\in \mathcal{S}$ with the following two properties:
\begin{itemize}
    \item  $A_{s}$ has nonzero probability measure, i.e., $\mathbb{P}[\{s\in A_{s}\}|\{\bm{G}_t=\bm{G}'_t\}]>0$ for any $\bm{G}'_t \in \mathcal{G}$.
    
    \item $A_{s}$ cannot be expressed as $B_{\bm{s}^c} \times \mathcal{S}^{u}$ for any $B_{\bm{s}^c} \subset \mathcal{S}^{c}$.
\end{itemize}
$\exists \bm{G}_1,\bm{G}_2 \in \mathcal{G}$, such that 
\begin{equation}
    \int_{s\in A_{s}} P(s|\bm{G}_1)ds \neq \int_{s\in A_{s}} P(s|\bm{G}_2)ds .
\end{equation}
By modeling the data generation process, $\bm{s}^c$ is block-wise identifiable.
\end{lemma2}

\begin{proof}
We split the proof into four steps for better understanding.

In Step1, we leverage the properties of the data generation process and the marginal distribution matching condition to express the marginal invariance with the indeterminacy transformation $\overline{h}:\mathcal{S} \rightarrow \mathcal{S}$ between the estimated and the true latent variables. The introduction of $\overline{h}(\cdot)$ allows us to formalize the block-wise identifiability condition.

In Step 2 and Step 3, we show that the estimated context latent variables $\hat{\bm{s}}^c$ do not depend on the true user preference variables $\bm{s}_{t}^u$, that is, $\overline{h}(\bm{s}^c)$ does not depend on the input $\bm{s}_{t}^u$. To this end, in Step 2, we derive its equivalent statements which can ease the rest of the proof and avert technical issues (e.g., sets of zero probability measures). In Step 3, we prove the equivalent statement by contradiction. Specifically, we show that if $\hat{\bm{s}}^c$ depends on $\bm{s}_{t}^u$, the invariance derived in Step 1 would break.

In Step 4, we use the conclusion in Step 3, the smooth and bijective properties of $h(\cdot)$, and the conclusion in Theorem \ref{the1}, to show the invertibility of the indeterminacy function between the context variables, i.e., the mapping $\hat{\bm{s}}^c=\overline{h}_{c}(\bm{s}^c)$ being invertible.

\noindent\textbf{Step 1}. According to the data generation process in Figure \ref{fig:model}(a), $\bm{s}^c$ is independent of $\bm{G}_t$, it follows that for any $A_{\bm{s}^c} \subseteq \mathcal{S}^{c}$,
\begin{equation}
\label{equ:lam1}
\begin{split}
    \mathbb{P}[\{\hat{g}^{-1}_{1:n_{c}}(\hat{y}|a_t,z_t) \in A_{\bm{s}^c}\}|\{\bm{G}_t=\bm{G}_t^1\}]&=\mathbb{P}[\{\hat{g}^{-1}_{1:n_{c}}(\hat{y}|a_t,z_t) \in A_{\bm{s}^c}\}|\{\bm{G}_t=\bm{G}_t^2\}], \forall \bm{G}_t^1,\bm{G}_t^2 \in \mathcal{G}\\ &\Longleftrightarrow\\
    \mathbb{P}[\{\hat{y}_t \in (\hat{g}^{-1}_{1:n_{c}})^{-1}(A_{\bm{s}^c}|a_t,z_t)\}|\{\bm{G}_t=\bm{G}_t^1\}]&=\mathbb{P}[\{\hat{y}_t \in (\hat{g}^{-1}_{1:n_{c}})^{-1}(A_{\bm{s}^c}|a_t,z_t)\}|\{\bm{G}_t=\bm{G}_t^2\}], \forall \bm{G}_t^1,\bm{G}_t^2 \in \mathcal{G},
\end{split}
\end{equation}
where $\hat{g}^{-1}_{1:n_{c}}: \mathcal{Y}\rightarrow \mathcal{S}^{c}$ denotes the estimated transformation from the reward to the context latent variables and $(\hat{g}^{-1}_{1:n_{c}})^{-1}(A_{\bm{s}^c}|a_t,z_t) \subseteq \mathcal{Y}$ is the preimage set of $A_{\bm{s}^c}$, that is, the set of estimated reward $\hat{y}_t$ originating from context variables $\hat{\bm{s}}^c$ in $A_{\bm{s}^c}$.

Because of the matching observation distributions between the estimated model and the true model in Equation (\ref{equ:gen}), the relation in Equation (\ref{equ:lam1}) can be extended to observation $y_t$ from the true generating process, i.e.,
\begin{equation}
\label{equ:lam2}
\begin{split}
    \mathbb{P}[\{y_t \in (\hat{g}^{-1}_{1:n_{c}})^{-1}(A_{\bm{s}^c}|a_t,z_t)\}|\{\bm{G}_t=\bm{G}_t^1\}]&=\mathbb{P}[\{y_t \in (\hat{g}^{-1}_{1:n_{c}})^{-1}(A_{\bm{s}^c}|a_t,z_t)\}|\{\bm{G}_t=\bm{G}_t^2\}], \\&\Longleftrightarrow\\
    \mathbb{P}[\{\hat{g}^{-1}_{1:n_{c}}(\hat{y}|a_t,z_t))^{-1}\in A_{\bm{s}^c}\}|\{\bm{G}=\bm{G}_t^1\}]&=\mathbb{P}[\{\hat{g}^{-1}_{1:n_{c}}(\hat{y}|a_t,z_t))^{-1}\in A_{\bm{s}^c}\}|\{\bm{G}=\bm{G}_t^2\}].
\end{split}
\end{equation}

Since $g$ and $\hat{g}$ are smooth and injective, there exists a smooth and injective $\overline{h}=\hat{g}^{-1} \circ g:\mathcal{S}\rightarrow\mathcal{S}$. We note that by definition $\overline{h}=h^{-1}$ where $h$ is introduced in the proof of Theorem \ref{the1}. Expressing $\hat{g}^{-1}=\overline{h}\circ g^{-1}$ and $\overline{h}_{c}:\mathcal{S}\rightarrow \mathcal{S}^c$ in Equation (\ref{equ:lam2}) yields
\begin{equation}
\label{equ:lam3}
\begin{split}
    \mathbb{P}[\{\overline{h}_{c}(s)\in A_{\bm{s}^c}\}|\{\bm{G}=\bm{G}_t^1\}]&=\mathbb{P}[\{\overline{h}_{c}(s)\in A_{\bm{s}^c}\}|\{\bm{G}=\bm{G}_t^2\}]\\&\Longleftrightarrow\\\mathbb{P}[\{\bm{s}^c\in \overline{h}_{c}^{-1}(s)\in A_{\bm{s}^c}\}|\{\bm{G}=\bm{G}_t^1\}]&=\mathbb{P}[\{\bm{s}^c\in \overline{h}_{c}^{-1}(s)\in A_{\bm{s}^c}\}|\{\bm{G}=\bm{G}_t^2\}]\\&\Longleftrightarrow \\\int_{s\in \overline{h}^{-1}_{c}(A_{\bm{s}^c})}p(s|\bm{G}_t^1)ds&=\int_{s\in \overline{h}^{-1}_{s}(A_{\bm{s}^c})}p(s|\bm{G}_t^2)ds,
\end{split}
\end{equation}
where $\overline{h}_{c}^{-1}(A_{\bm{s}^c})=\{s\in \mathcal{S}: \overline{h}_c(s)\in A_{\bm{s}^c}\}$ is the preimage of $A_{\bm{s}^c}$, i.e., those latent variables containing context variables in $A_{\bm{s}^c}$ after indeterminacy transformation $h$.

Based on the generating process in Equation \ref{equ:gen}, we can rewrite Equation (\ref{equ:lam3}) as follows:
\begin{equation}
\label{equ:lam4}
    \forall A_{\bm{s}^c} \subseteq \mathcal{S}^c, \int_{[{s_{t}^u}^{\top},{\bm{s}^c}^{\top}]^{\top} \in \overline{h}_c^{-1}(A_{\bm{s}^c})}p(\bm{s}^c)\left(p(\bm{s}_{t}^u|\bm{G}_t^1)-p(\bm{s}_{t}^u|\bm{G}_t^2)\right)d\bm{s}^c d\bm{s}_{t}^u = 0.
\end{equation}

\noindent\textbf{Step 2}. In order to show the block-identifiability of $\bm{s}^c$, we would like to prove that $\bm{s}^c:=\overline{h}_c([{\bm{s}^c}^{\top},{s_{t}^u}^{\top}]^{\top})$ does not depend on $s_{t}^u$. To this end, we first develop one equivalent statement (i.e., Statement 3 below) and prove it later step instead. By doing so, we are able to leverage the full-supported density function assumption to avert technical issues.
\begin{itemize}
    \item \textbf{Statement 1}: $\overline{h}_c([{\bm{s}^c}^{\top},{s_{t}^u}^{\top}]^{\top})$ does not depend on $\bm{s}^u$.
    \item \textbf{Statement 2}: $\forall \bm{s}^c\in \mathcal{S}^c$, it follows that $\overline{h}_c^{-1}(\bm{s}^c)=B_{\bm{s}^c}\times \mathcal{S}^{u}$, where $B_{\bm{s}^c}\neq \emptyset$ and $B_{\bm{s}^c} \subseteq \mathcal{S}^c$.
    \item \textbf{Statement 3}: $\forall \bm{s}^c \in \mathcal{S}^c, r\in \mathbb{R}^+$, it follows that $\overline{h}_c^{-1}(\mathcal{B_r}(\bm{s}^c))=B_{\bm{s}^c}^+\times\mathcal{S}^u$, where $\mathcal{B}_r(\bm{s}^c):=\{{\bm{s}^c}'\in \mathcal{S}^c: ||{\bm{s}^c}'-\bm{s}^c||^2<r, B_{\bm{s}^c}^+\neq \emptyset\}$, and $B_{\bm{s}^c}^+\subseteq \mathcal{S}^c$.
\end{itemize}

Statement 2 is a mathematical formulation of Statement 1. Statement 3 generalizes singletons $\bm{s}^c$ in Statement 2 to open, non-empty balls $\mathcal{B}_r(\bm{s}^c)$. Later, we use Statement 3 in Step 3 to show the contraction to the Equation (\ref{equ:lam4}).

Using the continuity of $\overline{h}_c(\cdot)$, we can show the equivalence between Statement 2 and Statement 3 as follows. We first show that Statement 2 implies Statement 3. $\forall \bm{s}^c \in \mathcal{S}^c, r\in \mathbb{R}^+, \overline{h}_c^{-1}({\bm{s}^c}')=B_{\bm{s}^c}'\times\mathcal{S}^{u}$, thus the union $\overline{h}_c^{-1}(\mathcal{B}_r(\bm{s}^c))$ also satisfies this property, which is Statement 3.

Then we show that Statement 3 implies Statement 2 by contradiction. Suppose that Statement 2 is false, then $\exists \hat{\bm{s}}^c \in \mathcal{S}^c$ such that there exist $\hat{s}^{c,B}\in\{s_{1:n_c}: s\in\overline{h}_c^{-1}(\hat{s}_c)\}$ and $\hat{s}^{u,B} \in \mathcal{S}^u$ resulting in $\overline{h}_c(\hat{s}^B)\neq \hat{s}_c$, where $\hat{s}^B=[(\hat{s}^{c,B})^{\top},(\hat{s}^{u,B})^{\top}]^{\top}$. As $\overline{h}_c(\cdot)$ is continuous, there exists $\hat{r}\in\mathbb{R}^{+}$ such that $\overline{h}_c(\hat{s}^B)\notin \mathcal{B}_{\hat{r}}(\hat{\bm{s}}^c)$. That is, $\hat{s}^B \notin h_c^{-1}(\mathcal{B}_{\hat{r}}(\hat{\bm{s}}^c))$. Also, Statement 3 suggests that $h_c^{-1}(\mathcal{B}_{\hat{r}}(\hat{\bm{s}}^c))=\hat{B}_{\bm{s}^c}\times \mathcal{S}^u$. By definition of $\hat{s}^B$, it is clear that $\hat{s}_{1:n_c}^B \in \hat{B}_{\bm{s}^c}$. The fact that $\hat{z}^B \notin h_c^{-1}(\mathcal{B}_{\hat{r}}(\hat{s}^{c}))$ contradicts Statement 3. Therefore, Statement 2 is true under the premise of Statement 3. We have shown that Statement 3 implies State 2. In summary, Statement 2 and Statement 3 are equivalent, and therefore proving Statement 3 suffices to show Statement 1.

\noindent \textbf{Step 3}. In this step, we prove Statement 3 by contradiction. Intuitively, we show that if $\overline{h}_c(\cdot)$ depended on $\hat{s}^u$, the preimage $\overline{h}_c^{-1}(\mathcal{B}_r(\bm{s}^c))$ could be partitioned into two parts (i.e. $B_s^*$ and $\overline{h}_c^{-1}(A_{\bm{s}^c}^*) \ B_s^*$ defined below.) The dependency between $\overline{h}_c(\cdot)$ and $\hat{s}^u$ is captured by $B_s^*$, which would not emerge otherwise. In contrast, $\overline{h}_c^{-1}(A_{\bm{s}^c}^*) \ B_s^*$ also exists when $\overline{h}_c(\cdot)$ does not depend on $\hat{s}^u$. We evaluate the invariance relation Equation (\ref{equ:lam4}) and show that the integral over $\overline{h}_c^{-1}(A_{\bm{s}^c}^*) \ B_{s}^*$ (i.e., $T_1$) is always 0, however, the integral over $B_s^*$ (i.e., $T_2$) is necessarily non-zero, which leads to the contraction with Equation (\ref{equ:lam4}) and thus shows the $\overline{h}_c(\cdot)$ cannot depend on $\hat{s}^u$.

First, note that because $\mathcal{B}_r(\bm{s}^c)$ is open $\overline{h}_c(\cdot)$ is continuous, the preimage $\overline{h}_c^{-1}(\mathcal{B}_r(\bm{s}^c))$ is open. In addition, the continuity of $h(\cdot)$ and the matched observation distributions $\forall \bm{G}'\in \mathcal{G}, \mathbb{P}[\{y\in A_y\}|\{\bm{G}=\bm{G}'\}]=\mathbb{P}[\{\hat{y}\in A_y\}|\{\bm{G}=\bm{G}'\}]$ lead to $h(\cdot)$ being bijection as shown in \cite{klindt2020towards}, which implies that that $\overline{h}_c^{-1}(\mathcal{B}_r(\bm{s}^c))$ is non-empty. Hence, $\overline{h}_c^{-1}(\mathcal{B}_r(\bm{s}^c))$ is both non-empty and open. Suppose that $\exists A_{\bm{s}^c}^*:=\mathcal{B}_{r^*}:=\mathcal{B}_{r^*}({\bm{s}^c}^*)$ where ${\bm{s}^c}^* \in \mathcal{S}^c, r^* \in \mathbb{R}^+$, such that ${B^s}^*=\{s\in \mathcal{S}: s\in \overline{h}_c^{-1}(A_{\bm{s}^c}^*), \{s_{1:n_c}\}\times \mathcal{S}^u\not\subseteq\overline{h}_c^{-1}(A_{\bm{s}^c}^*)\}\neq\emptyset$. Intuitively, ${B^s}^*$ contains the partition of the preimage $\overline{h}_c^{-1}(A_{\bm{s}^c}^*)$ that the user preference part $s_{n_c+1:n}$ cannot take on any value in $\mathcal{S}^u$. Only certain values of the style part were able to produce specific outputs of indeterminacy $\overline{h}_c(\cdot)$. Clearly, this would suggest that $\overline{h}_c(\cdot)$ depends on $\bm{s}^c$.  To show contraction with Equation (\ref{equ:lam4}), we evaluate the \textbf{LHS} of Equation (\ref{equ:lam4}) with such a $A_{\bm{s}^c}^*$.
\begin{equation}
\begin{split}
&\int_{[{\bm{s}^c}^{\top},{\bm{s}^u}^{\top}]^{\top}\in \overline{h}_c^{-1}(A_{\bm{s}^c}^*)} p(\bm{s}^c)\left(p(\bm{s}^u|\bm{G}_1)-p(\bm{s}^u|\bm{G}_2)\right)d\bm{s}^u d\bm{s}^c\\=&\underbrace{\int_{[{\bm{s}^c}^{\top},{\bm{s}^u}^{\top}]^{\top}\in \overline{h}_c^{-1}(A_{\bm{s}^c}^*) / {B^s}^*} p(\bm{s}^c)\left(p(\bm{s}^u|\bm{G}_1)-p(\bm{s}^u|\bm{G}_2)\right)d\bm{s}^u d\bm{s}^c }_{T_1} \\ &+ \underbrace{\int_{[{\bm{s}^c}^{\top},{\bm{s}^u}^{\top}]^{\top}\in {B^s}^*} p(\bm{s}^c)\left(p(\bm{s}^u|\bm{G}_1)-p(\bm{s}^u|\bm{G}_2)\right)d\bm{s}^u d\bm{s}^c}_{T_2}.
\end{split}
\end{equation}
We first look at the value of $T_1$. When $\overline{h}_c^{-1}(A_{\bm{s}^c}^*) /\ {B^{s}}^* = \empty$, $T_1$ evaluates to 0. Otherwise, by definition, we can rewrite $\overline{h}_c^{-1}(A_{\bm{s}^c}^*) / {B^s}^*$ as $C_{\bm{s}^c}^* \times \mathcal{S}^u$ where $C_{\bm{s}^c}^*\neq \empty$ and $C_{\bm{s}^c}^* \subset \mathcal{S}^c$. With this expression, it follows that 
\begin{equation}
\begin{split}
&\int_{[{\bm{s}^c}^{\top},{\bm{s}^u}^{\top}]^{\top}\in C_{\bm{s}^c}^*} p(\bm{s}^c)\left(p(\bm{s}^u|\bm{G}_1)-p(\bm{s}^u|\bm{G}_2)\right)d\bm{s}^c d\bm{s}^u\\=&\int_{\bm{s}^c\in C_{\bm{s}^c}^*} p(\bm{s}^c)\int_{\bm{s}^u\int\mathcal{S}^u}\left(p(\bm{s}^u|\bm{G}_1)-p(\bm{s}^u|\bm{G}_2)\right)d\bm{s}^c d\bm{s}^u\\=&\int_{\bm{s}^c\in C_{\bm{s}^c}^*} p(\bm{s}^c)(1-1)d\bm{s}^c=0.
\end{split}
\end{equation}
Therefore, in both cases, $T_1$ evaluates to 0 for $A_{\bm{s}^c}^*$.

Now, we address $T_2$. As discussed above, $\overline{h}_c^{-1}(A_{\bm{s}^c}^*)$ is open and non-empty. Because of the continuity of $\overline{h}_c(\cdot)$, $\forall s^{B}\in {B^{s}}^*$, there exists $r(s^{B}) \in \mathbb{R}^+$ such that $
\mathcal{B}_{r(s^B)}(s^B) \subseteq {B^s}^*$. As $p(s|\bm{G})>0$ over $(s,\bm{G})$, we have $\mathbb{P}[\{s\in {B^S}^*\}|\{\bm{G}=\bm{G}'\}]\geq \mathbb{P}[\{s \in \mathcal{B}_{r(s^B)}(s^B)\}|\{\bm{G}=\bm{G}'\}]>0$ for any $\bm{G}\in\mathcal{G}$. Assumption A4 indicates that $\exists \bm{G}_t^{1*},\bm{G}_t^{2*}$, such that 
\begin{equation}
T_2:=\int_{[{\bm{s}^c}^{\top},{s^{u}}^{\top}]^{\top}\in B_s^*}p({\bm{s}^c})\left(p({\bm{s}^u}|\bm{G}_t^{1*})-p({\bm{s}^u}|\bm{G}_t^{2*})\right)d\bm{s}^c d\bm{s}^u \neq 0.
\end{equation}
Therefore, for such $A_{\bm{s}^c}^*$, we would have $T_1+T_2\neq 0$, which leads to contradiction with Equation (\ref{equ:lam4}). We have proved by contradiction that Statement 3 is true and hence Statement 1 holds, that is, $\overline{h}_c(\cdot)$ does not depend on the user preference latent variables $s_u$.

\noindent \textbf{Step 4}. With the knowledge that $\overline{h}_c(\cdot)$ does not depend on the user preference latent variables $s^{u}$, we now show that there exists an invertible mapping between the true context variables $\bm{s}^c$ and the estimated version $\hat{s}_c$. As $\overline{h}(\cdot)$ is smooth over $\mathcal{S}$, its Jacobian can be written as:
\begin{equation}
\begin{gathered}
    {\bm{J}}_{\overline{h}}=\begin{bmatrix}
    \begin{array}{c|c}
        \textbf{A}:=\frac{\partial \hat{\bm{s}}^c}{\partial \bm{s}^c} & \textbf{B}:=\frac{\partial {\hat{s}}^c}{\partial \bm{s}^u} \\ \hline
        \textbf{C}:=\frac{\partial {\hat{s}}^u}{\partial \bm{s}^c} & \textbf{D}:=\frac{\partial {\hat{s}}^u}{\partial \bm{s}^u},
    \end{array}
    \end{bmatrix}
\end{gathered}
\end{equation}
where we use notation $\hat{\bm{s}}^c=\overline{h}(s)_{1:n_c}$ and $\hat{s}^u=\overline{h}(s)_{n_c+1:n}$. As we have shown $\hat{\bm{s}}^c$ does not depend on the user preference latent variables $\bm{s}^u$, it follows $\bm{B}=0$. On the other hand, since $h(\cdot)$ is invertible over $\mathcal{S}$, $\bm{J}_{\overline{h}}$ is non-singular. Therefore, $\bm{A}$ must be non-singular due to $\bm{B}=0$. We note that $\bm{A}$ is the Jacobian of the function $\overline{h}_c'(\bm{s}^c):=\overline{h}_c(s):\mathcal{S}^c\rightarrow\mathcal{S}^c$, meaning that $\bm{s}^c$ is component-wise identification.
\end{proof}


\bibliographystyle{plain}
\bibliography{cas-refs}


\end{document}